\documentclass[format=acmsmall, review=false, screen=true]{acmart}

\usepackage{booktabs} 

\usepackage[ruled]{algorithm2e} 

\SetAlFnt{\small}
\SetAlCapFnt{\small}
\SetAlCapNameFnt{\small}
\SetAlCapHSkip{0pt}
\IncMargin{-\parindent}

\setcopyright{none}
\begin{document}
\title[Makespan Minimization via Posted Prices]{Makespan Minimization via Posted Prices}
\author{Michal Feldman}
\authornote{Blavatnik School of Computer Science, Tel Aviv University, email: \texttt{mfeldman@tau.ac.il}.}
\affiliation{%
  \institution{Tel Aviv University}
  \city{Tel Aviv}
  \country{Israel}
  \department{Blavatnik School of Computer Science}}
\email{mfeldman@tau.ac.il}
\author{Amos Fiat}
\authornote{Blavatnik School of Computer Science, Tel Aviv University, email: \texttt{fiat@tau.ac.il}.}
\affiliation{%
  \institution{Tel Aviv University}
  \city{Tel Aviv}
  \country{Israel}
  \department{Blavatnik School of Computer Science}}
\email{fiat@tau.ac.il}
\author{Alan Roytman}
\authornote{Department of Computer Science, University of Copenhagen, email: \texttt{alanr@di.ku.dk}.}
\affiliation{%
  \institution{University of Copenhagen}
  \department{Department of Computer Science}
  \city{Copenhagen}
  \country{Denmark}
}
\email{alanr@di.ku.dk}

\begin{abstract}
We consider job scheduling settings, with multiple machines, where jobs arrive online
and choose a machine selfishly so as to minimize their cost. Our objective is the classic
makespan minimization objective, which corresponds to the completion time of the last job to complete.
The incentives of the selfish jobs may lead to poor performance. To reconcile the differing objectives,
we introduce posted machine prices. The selfish job seeks to minimize the sum of its completion time on
the machine and the posted price for the machine.
Prices may be static (i.e., set once and for all before any arrival) or dynamic (i.e., change over time), but they are determined only by the past, assuming nothing about upcoming events.
Obviously, such schemes are inherently truthful.

We consider the competitive ratio: the ratio between the makespan achievable by the pricing scheme and that of the optimal algorithm.
We give tight bounds on the competitive ratio for both dynamic and static pricing
schemes for identical, restricted, related, and unrelated machine settings.
Our main result is a dynamic pricing scheme for related machines that
gives a constant competitive ratio, essentially matching the competitive ratio of online algorithms for this setting.
In contrast, dynamic pricing gives poor performance for unrelated machines. This lower bound
also exhibits a gap between what can be achieved by pricing versus what can be achieved by online algorithms.
\end{abstract}

%
%
%

%
%


\thanks{
The work of M. Feldman and A. Roytman was partially supported by the European Research Council under the European Union's Seventh Framework Programme (FP7/2007-2013) / ERC grant agreement number 337122.
The work of A. Roytman was also partially supported by Thorup's Advanced Grant DFF-0602-02499B from the Danish Council for Independent Research.
}

\maketitle


\section{Introduction}

Online algorithms can be viewed as follows: events arrive over time, and upon the arrival of an event, the algorithm makes a decision, based only on the prior and current event, without knowing future events. There is a function that maps outcomes to costs or benefits, where the goal of such an algorithm is either to minimize costs or to maximize benefits.  The competitive ratio of an online algorithm seeks to compare the outcome of the online algorithm with the optimal outcome. The term competitive analysis was coined in \cite{karlin1988} and gives the ratio between the outcome (cost or benefit) achieved by an online algorithm and the outcome of an offline optimal solution. Herein we only deal with cost problems so this ratio is $\geq 1$.

We consider a setting where every online event is associated with a selfish agent. The agents have some associated true type that describes the nature of the event. Agents have some utility (or disutility) associated with the outcome. For many online algorithms, the decisions made by the algorithm might not be in the best interest of the agent. This may result in the agent misrepresenting her type so as to achieve a better outcome for herself.

The design of {\em mechanisms} mitigates the problem of strategic behavior described above.
In a mechanism, agents report their type, and the mechanism decides upon an outcome and upon payments, where payments are used to align the incentives of the agents with that of the mechanism.
A mechanism is truthful if it is always in the best interest of any agent to report her type truthfully.
In online settings, agents arrive sequentially, and the mechanism decides upon an outcome and payment for each arriving agent as they arrive\footnote{Such online mechanisms are called {\sl prompt} in that outcomes and payments are determined immediately, non-prompt online mechanisms have also been studied where, for example, the payment is only determined later.}.

With few notable exceptions (e.g., \cite{DBLP:conf/stoc/NisanR99}), problems studied in mechanism design (both online and offline settings) involve {\sl maximizing} social welfare or revenue.
Although optimal solutions to maximization/minimization objectives can be cast as the other, the competitive ratio is quite different in the two settings. Online algorithms have been devised with both maximization and minimization objectives. The technique of ``classify and randomly select'' \cite{ABFR94} often gives simple randomized algorithms for maximization objectives which also naturally translate into truthful mechanisms. In contrast, minimization objectives (e.g., $k$-server and makespan) require entirely different techniques. Converting online algorithms into mechanisms without performance degradation opens up an entire new class of problems for which incentive compatible mechanism design is applicable.

\vspace{2mm}
{\bf Dynamic posted prices.}
We consider truthful online mechanisms that take the form of {\em dynamic posted prices}. Dynamic pricing schemes are truthful online mechanisms that post prices for every possible outcome, before the next agent arrives. Then, the agent chooses the preferred outcome ---
minimizing the cost for the outcome minus the price tag associated with the outcome.

Such a mechanism is inherently truthful, since prices are determined irrespective of the type of the next agent.
Posted price mechanisms have many additional advantages over arbitrary truthful online mechanisms. In particular such mechanisms are simple~\cite{HR09}: agents need not trust nor understand the logic underlying the truthful mechanism, agents are not required to reveal their type, and there is no need to verify that the agents indeed follow the decision made by the truthful online mechanism.\footnote{One may suspect that any online mechanism gives rise to dynamic pricing schemes. This is not quite true: the online mechanism must be prompt, and, moreover, the online mechanism may require that ties (equal utility choices) be broken in a particular manner, and as a function of the agent type. In contrast, with dynamic pricing schemes agents may break ties arbitrarily. Many thanks to Moshe Babaioff,  Liad Blumrosen, Yannai A. Gonczarowski, and Noam Nisan for discussions clarifying this point. Clearly, any dynamic pricing scheme gives rise to a prompt online truthful mechanism with the same performance guarantees.}

 A posted price mechanism is a truthful online algorithm, and as such, can perform no better than the best online algorithm.
Our main goal in this paper is to study the performance of dynamic posted price mechanisms (quantified by the competitive ratio measure) and compare them with the performance of the best online algorithm.
One may think of this problem as analogous to one of the central questions in algorithmic mechanism design in offline settings: compare the performance of the best truthful mechanism (quantified by the approximation ratio measure) with the performance of the best non-truthful algorithm.

\vspace{2mm}
{\bf Makespan minimization in job scheduling.}
In this paper we study the design of online mechanisms for makespan minimization in job scheduling.
Events represent jobs, the job type contains the job's processing times on various machines.
Agents seek to complete their job as soon as possible, and therefore prefer to be assigned to a machine whose
load (including the new job) is minimized\footnote{In this interpretation ``load'' is the time required by the
server to deal with all current jobs in the server queue, and jobs are processed in a first-in-first-out manner,
i.e., jobs enter a server queue. In some papers ``load'' is used in the context of round-robin processing.}.
For simplicity of exposition, we assume that all jobs arrive (sequentially) at time zero. However, our positive
results hold even if jobs arrive at arbitrary times. Clearly, adding options (arbitrary arrival times) does
not invalidate impossibility results.  Existing online algorithms for the problem (e.g.,~\cite{AAFPW93}) are not truthful;
in that a job may misrepresent its size so as to get a preferential assignment to a machine.

An online truthful mechanism for this setting determines an allocation and payment for each arriving agent upon arrival.
That is, upon the arrival of a job, based on the job's processing times, the mechanism assigns the job to some machine and determines the payment the agent should make.
The cost of an agent is the sum of the machine's load (including her own processing time) and the payment.
Each agent seeks to minimize her cost.

A dynamic posted price mechanism for this setting sets prices on each machine, {\em before} the next agent arrives (prices may change over time).
The next agent to arrive seeks to minimize her cost, i.e., the load on the chosen machine (including her own load) plus the posted price on the machine. The agent breaks ties arbitrarily.



We consider this question for the goal of makespan minimization in job scheduling,
where the dynamic pricing scheme seeks to minimize the makespan, whereas selfish jobs seek to minimize their own completion time.
We assume FIFO processing within a machine, so the completion time of a job is the sum of the current load (prior to
the arrival of the job) plus the job's own processing time on the machine.  This problem has many applications,
including managing queues at banks, cloud computing settings where customers submit jobs and can lie about their processing
times, and crowdsourcing settings where taskmasters wish to hire workers to complete tasks while lying about
how long their task takes to complete.  In all such applications, we are interested in balancing loads appropriately.
To this end, we consider online makespan minimization for identical, restricted, related,
and unrelated machine models.

\vspace{2mm}
{\bf Examples.} To clarify the issue of selfish jobs, consider the following small toy problem: the setting is that of machines with speeds, machine $\#1$ has speed $1$, machine $\#2$ has speed $1/2$. There are also two jobs, job $a$ is of size $1/2$ and job $b$ is of size $1$. Clearly, the minimal makespan is achieved by assigning job $a$ to machine $\#2$ and job $b$ to machine $\#1$. This gives a makespan of one. Assume that the order of arrival is $a$, $b$. Job $a$ will prefer machine $\#1$ (completion time $1/2$) to machine $\#2$ (completion time $1$). Job $b$ will also prefer machine $\# 1$ (completion time $1.5$) to machine $\#2$ (completion time $2$).

In this specific case a static pricing of $1/2+\epsilon$ for machine $\#1$ and a price of zero for machine $\#2$ will result in the optimal makespan irrespective of the order of arrival of the jobs. If the order is $a$, $b$ then job $a$ prefers machine $\#2$ (completion time $1$ $+$ price $0$ $=$ $1$) over machine $\#1$ (completion time $1/2$ $+$ price $1/2+\epsilon$ $=$ $1+\epsilon$). The second job to arrive, job $b$, prefers machine $\#1$ (completion time $1$ $+$ price $1/2+\epsilon$ $=$ $1.5+\epsilon$) over machine $\#2$ (completion time $3$ $+$ price $0$ $=$ $3$). One can verify that the order $b,a$ will also achieve the same minimal makespan result.

The above example is somewhat misleading as we do not want to derive prices for a specific set of arriving jobs but for any arbitrary set and arbitrary order. In fact, we show that, in general, static prices are no better than a complete lack of prices (see Section~\ref{sec:static_pricing}), and only dynamic prices can guarantee a constant competitive ratio. A more detailed example that also illustrates the use of our dynamic pricing scheme for related machines (see Section \ref{sec:prices-related}) appears in Appendix \ref{app:example}.

\subsection{Our Model}

We have $m$ machines and $n$ jobs which arrive in an online manner.
Unrelated, related, restricted, and identical machine models are defined as follows:
\begin{enumerate}
\item For unrelated machines, the processing time of job $j$ on machine $i$ is given by $p_{ij}$.
\item In the related machines model,
each machine $i$ has some speed $s_i$ and each job $j$ has some associated size $p_j$.  The processing
time of job $j$ on machine $i$ is given by $p_{ij} = \frac{p_j}{s_i}$.
\item In the restricted machines model, each job $j$ has some associated size $p_j$.  The processing
time of job $j$ on machine $i$ is either $p_j$ or $\infty$.
\item In the identical machines model, each job $j$ has some associated size $p_j$, which is job $j$'s processing time on all machines.
\end{enumerate}
In the online setting, neither processing times, $p_{ij}$, nor size, $p_j$, are known until job $j$ arrives.
Machine speeds $s_i$ are known in advance.
While jobs do arrive in adversarial order, by renaming we can assume that job $j$ is the $j^{th}$ job to arrive.

We denote by $\sigma$ an input sequence consisting of jobs $1$ through $n$.
For a machine $i$, we let $M_i(j)$ denote the set
of jobs that have been assigned to machine $i$ after jobs $1$ through $j$ have arrived.  We denote the load
on machine $i$ after jobs $1$ through $j$ have arrived by $\ell_i(j) = \sum_{b \in M_i(j)}p_{ib}$.
For the makespan objective, the goal is to minimize $\max_i \ell_i(n)$.  Given an input sequence $\sigma$,
we denote by $L^*(\sigma)$ an optimal solution that is omniscient and knows the entire input sequence $\sigma$
in advance (i.e., an optimal solution that knows all jobs' processing times).  When clear from the context, we omit the parameter $\sigma$ and simply write $L^*$.

A \emph{dynamic pricing scheme} $D$, given an input sequence $\sigma$, outputs a sequence of $n$ vectors
$\pi_1,\ldots,\pi_n \in \mathbb{R}^m$, where each $\pi_j = \left(\pi_{1j},\ldots,\pi_{mj}\right)$
represents a vector of prices for each of the $m$ machines.  Each vector $\pi_j$ is determined before the $j^{th}$ job
arrives.  We view jobs as rational selfish agents who must choose the machine to which they wish to be assigned.  In particular, we model each arriving agent $j$'s
cost on machine $i$ as $c_{ij} = \ell_i(j-1) + p_{ij} + \pi_{ij},$ where $\ell_i(j-1)$ is the load on machine $i$ before $j$ arrives,
$p_{ij}$ is the processing time of job $j$ on machine $i$, and $\pi_{ij}$ is the price on machine $i$ (determined before $j$ arrives).
Hence, agent $j$'s cost on machine $i$ represents how long agent $j$ must wait in order to be processed by machine $i$, given
the load of the machine upon $j$'s arrival, plus some price amount determined by the dynamic pricing scheme.
We assume that agents are rational and wish to minimize their cost.  That is, agent $j$ chooses a machine that attains the minimum
value $\min_i c_{ij}$.

Note that, in our model, each player is a job, not a machine.  Hence, a job (i.e.,
player) may potentially misreport its processing times to the scheme in order to lower its incurred cost.  However, dynamic
pricing schemes are inherently truthful (since the prices are set independently of reported processing times),
and hence jobs never benefit from lying regarding their processing times on machines.
We denote by $D(\sigma)$ the makespan of the schedule produced by the dynamic pricing scheme $D$ given input $\sigma$, $D(\sigma)=\max_i \ell_i(n)$ --- the maximum load of any machine.

We can similarly define a \emph{static pricing scheme}, which simply sets one $m$-dimensional vector of prices $\pi$ (i.e., a single value for each machine)
before any agents arrive.  We do not permit static pricing schemes to change prices over time (so that $\pi_{i1} = \pi_{ij}$ for all $i$
and $j \geq 1$).  For this reason, when referring to static pricing schemes, we simply use one subscript instead of two.
In particular, we write $\pi_{i*}$ to denote the price on machine $i$ (at all times), so that $\pi = (\pi_{1*},\ldots,\pi_{m*})$.

We say that a dynamic pricing scheme is $c$-competitive if, given any input sequence $\sigma$, the scheme
always produces a schedule with a makespan satisfying $D(\sigma) \leq c \cdot L^*(\sigma) + a$ (assuming that agents behave selfishly),
where we allow some additive constant $a$.

\subsection{Our Contributions}

We give tight results for the competitive ratios that can be achieved via dynamic and static
pricing schemes for the problem of minimizing makespan.  We study identical, related, restricted, and unrelated machine models.
Our results, in comparison with previous work, are summarized in Table~\ref{tab:cr}.
Our main results are as follows ($m$ denotes the number of machines).

\begin{enumerate}
\item A dynamic pricing scheme that achieves an $O(1)$ competitive ratio for the related machines model.  This matches the $O(1)$-competitive result (of a non-truthful online algorithm) given in~\cite{AAFPW93}.

\item A lower bound on the competitive ratio of any dynamic pricing scheme of $\Omega(m)$ for unrelated machines.
Our lower bound holds for any randomized dynamic pricing scheme, even assuming an oblivious adversary.
\end{enumerate}

To the best of our knowledge,
the lower bound for unrelated machines exhibits the first gap between what can be achieved by dynamic pricing schemes versus what can be achieved by online algorithms.
That is, a gap of $\Omega(m)$ (for randomized dynamic pricing) versus $O(\log m)$ (achieved via deterministic online algorithms~\cite{AAFPW93}).

Our $O(1)$-competitive dynamic pricing scheme for related machines also holds in a more general model
where jobs arrive in real time (as opposed to arriving in sequence).  In such a setting, jobs are processed
over time on machines and are eventually removed from machines completely (upon being
fully processed).  The objective is to minimize the completion time of the last job to complete.

In addition, we show that static pricing schemes and the online greedy algorithm\footnote{We refer to the online greedy algorithm as the greedy algorithm that
assigns each job to a machine that minimizes the current load plus processing time of the job on the machine.} achieve the same performance up to a constant factor.
Clearly, the static pricing scheme that sets all prices to zero mimics the greedy algorithm.
Furthermore, we show that any lower bound on the competitive ratio of the greedy algorithm translates to the same lower
bound on any static pricing scheme\footnote{This result holds for any deterministic pricing scheme. For randomized schemes,
it holds as long as the lower bound for the greedy algorithm does not depend on the tie-breaking rule.} (up to constant factors),
for all machine models considered in this paper.  We note that such a greedy algorithm is
$O(1)$-competitive for identical machines~\cite{G66}, $\Theta(\log m)$-competitive for related machines~\cite{AAFPW93},
$\Theta(\log m)$-competitive for the restricted assignment model~\cite{ANR92}, and $\Theta(m)$-competitive for
unrelated machines~\cite{AAFPW93}.  These results appear in the columns labeled Greedy and Static Pricing in Table~\ref{tab:cr}.


\begin{table}%
\caption{Competitive ratio comparison of the greedy algorithm, the best online algorithm, static pricing schemes,
and dynamic pricing schemes.  Here, the greedy algorithm denotes the algorithm that assigns each job to the machine
that minimizes the current load plus processing time of the job on the machine. Results in the Static Pricing and Dynamic Pricing columns are from this paper.}
\label{tab:cr}
\begin{minipage}{\columnwidth}
\begin{center}
\begin{tabular}{lllll}
  \toprule
Machine Model & Greedy & Best Online & Static Pricing & Dynamic Pricing \\ \midrule
Identical & $O(1)$ \cite{G66} & $O(1)$ \cite{G66}& $O(1)$ & $O(1)$ \\
Related & $\Theta(\log m)$ \cite{AAFPW93}  & $O(1)$ \cite{AAFPW93} & $\Theta(\log m)$ & $O(1)$  \\
Restricted & $\Theta(\log m)$ \cite{ANR92} & $\Theta(\log m)$ \cite{ANR92} & $\Theta(\log m)$ & $\Theta(\log m)$ \\
Unrelated & $\Theta(m)$ \cite{AAFPW93} & $\Theta(\log m)$ \cite{AAFPW93,ANR92} & $\Theta(m)$ & $\Theta(m)$ \\
  \bottomrule
\end{tabular}
\end{center}
\end{minipage}
\end{table}%

\subsection{Techniques}

{\bf Positive Results for Related Machines.}
Our $O(1)$-competitive dynamic pricing scheme for the related machines model is inspired by the corresponding
related machines algorithm given in~\cite{AAFPW93}, referred to as Slow-Fit in \cite{AKPPW97}.


We now describe the main ideas behind our main result by discussing Slow-Fit.  We assume that machines are sorted
in increasing order of their speed, so that $s_1 \leq \cdots \leq s_m$.  In particular,
Slow-Fit operates in phases, where each phase maintains a lower bound $\Lambda$ on the current optimal solution.
The estimate $\Lambda$ doubles from phase to phase. A job $j$ is said to be feasible on machine $i$ if
$\ell_i(j-1) + \frac{p_j}{s_i} \leq 2 \Lambda$. Slow-Fit assigns the job to the lowest index (slowest) machine on which it is feasible.
If no machine is feasible, Slow-Fit doubles $\Lambda$.  This doubling process repeats until $\Lambda$ exceeds the value
of the optimal solution (i.e., $\Lambda$ becomes an upper bound), after which, for any incoming job, some machine is feasible (and
hence such jobs can be assigned).  Clearly, Slow-Fit depends
on the incoming job's size. The challenge in emulating Slow-Fit via a dynamic pricing scheme is that prices must be set before the size of the next job is revealed.

To show the underlying ideas, we now make several assumptions for which we show how to set prices for two
machines that perfectly emulate Slow-Fit. The assumptions are (a) Both machines have different speeds,
(b) Selfish jobs break ties in favor of machine $1$, and (c) $\Lambda$ is a known upper bound on the optimal solution.

Without loss of generality the price on machine $1$ is zero. Slow-Fit assigns job $j$ to machine $1$ if and only if
job $j$ is feasible on machine $1$, hence, we would like to set a price on machine $2$, $\pi_{2j}$, so that
$$\ell_1(j-1)+\frac{p_j}{s_1} \leq \ell_2(j-1)+\frac{p_j}{s_2} + \pi_{2j} \Longleftrightarrow \ell_1(j-1) + \frac{p_j}{s_1} \leq 2 \Lambda.$$
This is achieved by setting
$$\pi_{2j} = \ell_1(j-1)-\ell_2(j-1) + s_1\left(\frac{1}{s_1} - \frac{1}{s_2}\right)\left(2\Lambda - \ell_1(j-1)\right).$$
As the price $\pi_{2j}$ is independent of $p_j$ this gives a valid dynamic pricing scheme. Substitution and rearrangement show
that a job is assigned to machine 1 if and only if it is feasible on machine 1, as required. Note that this does not hold for equal speed
machines, and if tie-breaking is not in favor of machine 1.

It follows from the simple example above, that the following issues must be considered so as to construct a dynamic pricing scheme that attempts to emulate Slow-Fit, these are:
\begin{itemize}
    \item {\sl Equal Speed Machines}:
      Imagine that machines $i$ and $i+1$ have the same speed, and job $j$ is feasible on both,
then job $j$ should be scheduled on machine $i$. However, it may be that a larger job $j$ is infeasible
on $i$ but still feasible on $i+1$, in this case it should be assigned to $i+1$. However, if the machines
have the same speed, then irrespective of any prices and job size, the same machine will
always be chosen (as the difference in costs between machines $i$ and $i+1$ is constant).

  \item {\sl New Phase Recognition}: A new phase starts when the job has no feasible machine. However, the
pricing scheme cannot tell that a new phase is about to begin because it does not know the size of the next job to arrive.

  \item {\sl Machine Tie-Breaking}: Beyond the issue of feasibility (which is also a problem, see above),
as the job size increases, different machines (of different speeds) will attain the minimal cost,
irrespective of the prices and the current loads. Ergo, a job cannot be assumed to choose the lowest index machine.
  \end{itemize}

If none of these issues were to arise, then it would be possible to come up with a dynamic pricing scheme that would precisely
mimic the decisions made by Slow-Fit (as in the two machine example above).  To deal with these issues, we design a new online
algorithm, Flex-Fit, a variant of Slow-Fit (see Algorithm~\ref{alg:onrel}). Flex-Fit allows more flexibility in assigning jobs to multiple
machines, and in deciding when to start a new phase. This new algorithm does lend itself to dynamic pricing schemes with the
same competitive ratio, up to a constant factor.

The dynamic pricing scheme that emulates Flex-Fit is described in Algorithm \ref{alg:dyrel}.
We carefully choose a subset of machines that have strictly increasing speeds on which to place finite prices.
Other machines get a price of infinity. The prices are set such that the job prefers lower indexed machines
to higher indexed machines if and only if the job is feasible on the lower indexed machine.



\vspace{2mm}

{\bf Impossibility Results for Unrelated Machines.}
We next describe techniques used in our $\Omega(m)$ lower bound on the competitive ratio of dynamic
pricing for unrelated machines. As a warm up, we give a deterministic lower bound.  The lower bound
job sequence consists of two types of jobs, depending on two cases regarding the behavior of the deterministic dynamic pricing scheme.  In case 1,
we introduce a type 1 job that results in an increase in the sum of machines' loads, for the dynamic pricing scheme.
In contrast, any type 1 job is assigned in an optimal solution without an increase in any machine load.
In case 2, we introduce a type 2 job, which always chooses  machine $1$ (under the dynamic pricing scheme).
On the other hand, an optimal solution can always
spread out any sequence of $m$ type 2 jobs. This input sequence shows a gap of $\Omega(m)$ for the competitive ratio
of deterministic dynamic pricing schemes.

Our randomized lower bound holds against oblivious
adversaries (i.e., adversaries that must construct the entire input sequence in advance,
before seeing any coin flips of the algorithm).  To achieve our randomized lower bound,
we use the same two types of jobs as in the deterministic case.  We show
how to construct such a sequence obliviously, depending on the relative probability of being in case 1 or in case 2.

\vspace{2mm}

{\bf Static Pricing $\equiv$ Greedy.}
The non trivial direction is to show that every static
pricing scheme can be as bad as the greedy algorithm.  To do so, we observe
that static pricing schemes can be viewed as starting the online process with some initial (arbitrary) imbalance in the loads.  For deterministic static pricing schemes we show how to flatten out
the effective loads ($=$ load $+$ price) so that they are all equal. Once this is done,
we can then apply the greedy lower bound sequence, obtaining a similar lower bound
result for any static pricing scheme.

For randomized schemes, we give a different construction, that holds as long as the lower bound for the greedy algorithm does not depend on the tie-breaking rule. The idea is to blow up the job sizes in the greedy lower bound sequence so that the initial imbalance in the effective load becomes negligible.

\subsection{Related Work}



\paragraph{Online Algorithms}

Introduced in the context of self-adjusting search trees \cite{ST85a}, paging, list update \cite{,ST85b}, and snoopy caching \cite{KarlinMRS86}, there soon arose a vast host of online problems for which competitive analysis was applied. These include problems such as metrical task systems \cite{borodin92}, the $k$-server problem \cite{ManasseMS90,KP95}, scheduling problems, routing problems, and many more. One particular class of problems that has been widely studied is that of online makespan minimization~\cite{AAFPW93,ANR92,MRT13,ABFP13}.

\vspace{-2mm}

\paragraph{Online Makespan Minimization.}
The literature on online makespan minimization is vast, we only discuss the most relevant
works.  Online load balancing results for a variety of machine
models appear in~\cite{AAFPW93}.  An $O(\log m)$-competitive algorithm for unrelated machines and an $8$-competitive algorithm for related machines
were given (i.e., the Slow-Fit algorithm).  It was also shown that greedy is $\Theta(m)$-competitive for unrelated machines
and $\Theta(\log m)$-competitive for related machines.
The greedy algorithm was shown to be $\Theta(\log m)$-competitive for the restricted assignment model
(moreover, no online algorithm can do better)~\cite{ANR92}.  Results for the identical machines model appear in~\cite{BFKV92} and~\cite{A97}, where
a $(2-\epsilon)$-competitive online algorithm for a small fixed $\epsilon > 0$ and a $1.923$-competitive
algorithm were given, respectively.

Many other makespan minimization problems have been studied in the online setting, including different objectives
such as minimizing the $L_p$ norm for $p \geq 1$~\cite{AAGKKV95}
(the classic makespan minimization problem corresponds to minimizing the $L_{\infty}$ norm), settings where machines
have activation costs~\cite{ABFP13,MRT13}, and load balancing in the multidimensional setting~\cite{MRT13,IKKP15}.

\vspace{-2mm}

\paragraph{Static and Dynamic Pricing Schemes for Online Settings.}
Dynamic pricing schemes for a variety of problems appear in~\cite{CEFJ15}.  In particular,~\cite{CEFJ15} gave an
$O(k)$-competitive algorithm for the $k$-server problem on a line, an $O(m)$-competitive algorithm
for metrical task systems (where $m$ denotes the number of states), and a competitive ratio that is logarithmic in the
ratio of the maximum to minimum distances between points for metrical matching on a line.
Additional static and dynamic pricing schemes appear in~\cite{FMN08}, where queue management problems were studied, and constant competitive ratios were obtained for social welfare.
%
Dynamic pricing schemes were also considered in~\cite{KL03}, in which the revenue maximization problem where a seller has an unlimited supply
of identical goods was studied.
A dynamic pricing scheme for routing small jobs (relative to the edge capacities)
through a network was considered in \cite{AAM03}.




\vspace{-2mm}

\paragraph{Posted Prices for Social Welfare and Revenue.}
Posted pricing schemes~\cite{FGL15,CHK07,CHMS10}
need not be online, may use non-anonymous pricing, and often assume something is known about the future
(e.g., public valuations, Bayesian settings, etc.).
There is a large body of works on posted price mechanisms for social welfare and revenue maximization.
In the full information setting (only the order of arrival is unknown), a static posted pricing scheme was given that
obtains the optimal welfare for unit-demand buyers, and at least half of the optimal welfare for any valuation function~\cite{CEFF16}.
This uses ideas from \cite{FGL13}.
%
The Bayesian setting was considered in~\cite{FGL15}, where agents'
valuations are drawn from a product distribution over XOS valuations.
A static posted pricing scheme was given that achieves at least half the optimal welfare (in expectation).
A general framework for the design of posted price mechanisms for welfare maximization in Bayesian settings was devised in~\cite{DFKL16}.
Pricing schemes for revenue maximization in Bayesian settings were considered in \cite{CHK07,CHMS10,CMS10}.
In these settings agents arrive sequentially and are offered (non-anonymous) prices.
It was shown that the optimal revenue can be approximated to within a constant factor in various single-parameter and multi-parameter settings.

\vspace{-2mm}

\paragraph{Coordination Mechanisms for Job Scheduling.}
There is also a research agenda within the price of anarchy literature that studies the notion of coordination
mechanisms.  This body of work focuses on non-truthful mechanisms in the offline setting, where jobs are selfish
agents (note that jobs are also selfish agents in our work).  Here, performance is measured in terms of the price of anarchy.
A coordination mechanism for identical machines appears in~\cite{CKN04}, where it was shown that the price of anarchy
is $\frac{4}{3} - \frac{1}{3m}$.  Various local policies for a variety of machine models appear in~\cite{ILMS05}.
It was shown that any deterministic coordination mechanism has a price of anarchy of $O(\log m)$
for the related and restricted models, and an
$\Omega(\log m)$ lower bound was given for the price of anarchy for the restricted model.
For unrelated machines, a $\Theta(m)$ bound was given for a simple randomized policy.
The weighted sum of completion times objective has also been studied~\cite{ACH14}.
\vspace{-2mm}

\section{Pricing Related Machines}
\label{sec:prices-related}

We begin by giving an online algorithm for the load balancing problem on
related machines with a constant competitive ratio.  Our online algorithm
is inspired by the Slow-Fit algorithm~\cite{AAFPW93}.  We do this to aid us in designing a dynamic pricing
scheme that can mimic the behavior of the online algorithm.  This ensures that our dynamic pricing scheme
will have the same competitive ratio as the online algorithm.

We first assume that machines are sorted in increasing order of speed, so that
$s_1 \leq s_2 \leq \cdots \leq s_m$ (i.e., the first machine is the slowest machine,
and the $m^{th}$ machine is the fastest machine).  Our algorithm and dynamic pricing scheme proceed in phases, where each phase
depends on our current estimate $\Lambda$ of the optimal makespan $L^*$.  A new phase begins upon realizing that the current estimate
is too small, at which point we update the estimate accordingly.  We use the notion of virtual loads in
our algorithms, which we denote by $\hat{\ell}_i(j)$.  Virtual loads capture the real load within a particular phase,
and are reset to zero once a new phase begins.  The real load on a machine is essentially given by the sum
of virtual loads over all phases.  All loads and virtual loads begin at zero.

Our algorithms use the notion of representative machines, which are determined by the current virtual loads $\hat{\ell}_i(j)$.
\begin{definition}[Representative]
Fix any machine $i$ and job $j$.  Let $R = \{k : s_k = s_i\}$ be the set of machines with the same speed as $i$.
We say the representative of machine $i$ when job $j$ arrives is an arbitrary machine $k \in R$ that minimizes $\hat{\ell}_k(j-1)$.
We denote the representative of machine $i$ when $j$ arrives by $r_i(j)$.
\end{definition}
Note that, for any machine $i$, if $i$ is the only machine with a speed of $s_i$, then we have $r_i(j) = i$ (for all $j$).
Hence, the notion of a representative is mainly useful when there are multiple machines with the same speed.  In particular, the notion
of representatives enables us to choose one machine out of many that have the same speed (note that representatives may change over time
as jobs arrive).  Moreover, we always have the property that $s_{r_i(j)} = s_i$.

\subsection{Flex-Fit: A Variant of Slow-Fit}

\begin{algorithm}
\DontPrintSemicolon
Assign job $1$ to machine $r_m(j)$ \;
$\Lambda \leftarrow \frac{p_1}{s_m}$ \;
$\hat{\ell}_i(1) \leftarrow 0$ for all $i$ \;
\While{job $j$ arrives} {
	$T \leftarrow \{i : \hat{\ell}_i(j-1) + \frac{p_j}{s_i} \leq  (2+\epsilon) \Lambda\}$ \;
	\If{$T \neq \emptyset$} {
		$S \leftarrow \{i : \hat{\ell}_i(j-1) + \frac{p_j}{s_i} \leq  2 \cdot \Lambda\}$ \;
		\If{$S = \emptyset$} {
			Optionally call New-Phase($j$) and continue with the next job
		}
		$k \leftarrow$ minimum machine index in $S$ (set $k \leftarrow m$ if $S = \emptyset$)\;
		Assign $j$ to an arbitrary machine $r^* \in \{r_i(j) : s_i \leq s_k \textrm{ and } i \in T\}$ \;
		$\hat{\ell}_{r^*}(j) \leftarrow \hat{\ell}_{r^*}(j-1) + \frac{p_j}{s_i}$ \;
		$\hat{\ell}_{i'}(j) \leftarrow \hat{\ell}_{i'}(j-1)$ for all $i' \neq r^*$ \;
	}
	\Else{
		New-Phase($j$)
	}
}
\caption{Flex-Fit: A Variant of Slow-Fit\label{alg:onrel}}
\end{algorithm}

\begin{algorithm}
\DontPrintSemicolon
\SetKwInOut{Input}{input}
\Input{Job $j$}
	Assign $j$ to any machine with a speed of $s_m$ \;
	$\Lambda \leftarrow \max\{2, 2^{\left\lceil \log_2 \frac{p_j}{s_m \Lambda}\right\rceil}\} \cdot \Lambda$ \;
	$\hat{\ell}_i(j) \leftarrow 0$ for all $i$\;
\caption{New-Phase\label{alg:phase}}
\end{algorithm}

For any fixed $\epsilon > 0$, we give Algorithm Flex-Fit (which is inspired by the Slow-Fit algorithm in~\cite{AAFPW93}).
The main theorem we show in this section is that Flex-Fit is $O(1)$-competitive for the makespan minimization
problem on related machines.

We give some intuition for the algorithm.  Initially, we assign the first job to some fastest machine and obtain a lower
bound estimate of $L^*$, the optimal makespan.  The algorithm considers two sets:
$T = \{i : \hat{\ell}_i(j-1) + \frac{p_j}{s_i} \leq  (2+\epsilon) \Lambda\}$ and
$S = \{i : \hat{\ell}_i(j-1) + \frac{p_j}{s_i} \leq  2 \cdot \Lambda\}$ (note that $S \subseteq T$).
In particular, a job is feasible on machine $i$ if and only if $i \in S$.  The set $T$ consists of machines
that are slightly infeasible for job $j$.

If $T = \emptyset$, then the algorithm simply begins a new phase. Starting a new phase corresponds to assigning
the job $j$ to any machine of speed $s_m$, along with updating $\Lambda$
appropriately and resetting virtual loads.   If $S \neq \emptyset$, then the algorithm
finds the lowest indexed machine in $S$, namely machine~$k$.  In this case, it is free to assign $j$ to the
representative $r_i(j)$ of any machine $i$ such that $s_i \leq s_k$ and $i \in T$.
If $S = \emptyset$, then  Flex-Fit is allowed to take one of two options:
begin a new phase and then wait to process the next job, or assign $j$ to the representative $r_i(j)$
of any machine $i$ such that $i \in T$.  We need this flexibility in order to mimic the online algorithm
via a dynamic pricing scheme.  In particular, tie-breaking issues in our dynamic pricing scheme may arise,
where a job can be indifferent between choosing a machine of speed strictly less than $s_m$ (in which case
a new phase does not begin), or a machine of speed $s_m$ (in which case a new phase may possibly begin).
Since we do not have control over which machine agents choose, our online algorithm must be sufficiently flexible
and allow either of the two options to be taken.

Note that we assume job $1$ has $p_1 > 0$ (otherwise our initial estimate $\Lambda = 0$).  If this is not the case,
we simply wait until such a job $j$ arrives, and reindex jobs so that $j=1$.

To show the theorem, we first prove two lemmata, Lemma \ref{lem:phaseload} and Lemma \ref{lem:nonempty}.

Lemma \ref{lem:phaseload} says that, as long as $S \neq \emptyset$ when a particular job $j$ arrives, the load on the machine
to which $j$ is assigned is within a constant of our estimate $\Lambda$. We defer the proof of Lemma \ref{lem:phaseload} to Appendix \ref{app:proofs-related}.

\begin{lemma}\label{lem:phaseload}
For any job $j > 1$, if $S \neq \emptyset$, Flex-Fit always assigns job $j$ to a machine $i$ such that
$\hat{\ell}_i(j-1) + \frac{p_j}{s_i} \leq (2+\epsilon)\Lambda$.
\end{lemma}

Lemma \ref{lem:nonempty} says that once our estimate $\Lambda$ is at least the optimal makespan, then no more new phases
are initiated (the proof is similar to the proof that Slow-Fit is competitive~\cite{AAFPW93}, although it must be
adapted appropriately). We defer the proof of Lemma \ref{lem:nonempty} to Appendix \ref{app:proofs-related}.

\begin{lemma}\label{lem:nonempty}
If $\Lambda \geq L^*$, then $S \neq \emptyset$ for all jobs $j > 1$.
\end{lemma}

We now conclude the proof of the main theorem in this section.
\begin{theorem}\label{thm:ff}
Flex-Fit is $O(1)$-competitive for the makespan minimization problem on related machines.
\end{theorem}
\begin{proof}
Observe that the final makespan of the online algorithm is at most (a) the sum of the maximal loads over all phases,
plus (b) any additional processing times incurred due to explicitly assigning jobs to some fastest
machine. That is, (b) refers to the very first job, along with all jobs that cause a phase to end.

We let $\Lambda_1 = \frac{p_1}{s_m}$ (i.e., the
online algorithm's first estimate once job $1$ arrives), and in general define $\Lambda_h$ to be the value
of $\Lambda$ during phase $h \geq 1$.  Observe that the estimate $\Lambda_h$ is always of the form $2^{a_h} \cdot \Lambda_1$
for some integer $a_h \in \{0,1,2,\ldots\}$ (i.e., $\frac{\Lambda_h}{\Lambda_1}$ is always a power of $2$ for every phase $h$).  Moreover,
we have $\Lambda_{h+1} \geq 2 \cdot \Lambda_h$ (since its new value is always its old value,
multiplied by $\max\{2, 2^{\lceil \log_2 \frac{p_j}{s_m \Lambda_h}\rceil}\}$, where $j$ is the current job
that caused a new phase to begin).  Hence, $a_h < a_{h+1}$.

Let $k$ be the number of phases, where $k \geq 1$.  We first show that $\Lambda_k \leq 2 \cdot L^*$,
where $\Lambda_k$ denotes the value of $\Lambda$ during the final phase.
If there is only one phase, we
are done since the final estimate is $\Lambda_k = \Lambda_1 = \frac{p_1}{s_m} \leq L^* \leq 2 \cdot L^*$.  Now,
suppose there are $k \geq 2$ phases.  Consider the $(k-1)^{st}$ phase (i.e., the phase just before the last phase).
The reason why the algorithm ended phase $k-1$ and started phase $k$ is because the set $S$ was empty when some job
arrived.  By Lemma~\ref{lem:nonempty},
it must have been the case that the estimate $\Lambda_{k-1}$ was strictly less than $L^*$.  There are two cases:
$\Lambda_k = 2 \cdot \Lambda_{k-1}$ or $\Lambda_k = 2^{\lceil \log_2 \frac{p_j}{s_m \Lambda}\rceil} \cdot \Lambda_{k-1}$
(where $j$ is the job that caused phase $k$ to begin).
In the first case, we have $\Lambda_k = 2 \cdot \Lambda_{k-1} < 2 \cdot L^*$.  In the second case, we also have
$\Lambda_k = 2^{\lceil \log_2 \frac{p_j}{s_m \Lambda_{k-1}}\rceil} \cdot \Lambda_{k-1} \leq
2 \cdot \frac{p_j}{s_m \Lambda_{k-1}} \cdot \Lambda_{k-1} \leq 2 \cdot L^*$ (since $\frac{p_j}{s_m} \leq L^*$).

By Lemma~\ref{lem:phaseload}, the total load that accumulates during a phase $h$ (i.e., while $T \neq \emptyset$, or
possibly until $S = \emptyset$) is at most $(2 + \epsilon) \Lambda_h$.  In addition, the algorithm also assigns
jobs to a machine with the fastest speed (e.g., job $1$, and all jobs
that cause a new phase to begin).  Such jobs $j$ add at most $\frac{p_j}{s_m}$ to the makespan.  Note that the job $j$ which
causes some phase $h$ to begin satisfies
$\frac{p_j}{s_m} \leq 2^{\lceil \log_2 \frac{p_j}{s_m \Lambda_{h-1}}\rceil} \cdot \Lambda_{h-1} \leq \Lambda_h$ (recall
that $\frac{p_1}{s_m} = \Lambda_1$ by definition).
Hence, each phase $h$ incurs a makespan of at most $\Lambda_h + (2+\epsilon)\Lambda_h$.  Let $a_h \in \{0,1,2,\ldots\}$
be the integer such that $\Lambda_h = 2^{a_h}\cdot\Lambda_1$.  In all, the final makespan
is at most:
$$\sum_{h=1}^k(3+\epsilon)\Lambda_h = (3+\epsilon)\Lambda_1\sum_{h=1}^k2^{a_h} \leq (3+\epsilon)\Lambda_1\cdot2^{a_k + 1}
= 2\cdot(3+\epsilon)\cdot2^{a_k}\Lambda_1 \leq 4\cdot(3+\epsilon)\cdot L^*.  $$
This concludes the proof of the theorem.
\end{proof}

\subsection{Dynamic Pricing Scheme}

\begin{algorithm}
\DontPrintSemicolon
$\pi_{m1} \leftarrow 0$, $\pi_{i1} \leftarrow \infty$ for all $i < m$\;
After agent $1$ chooses a machine $q$: \;
\Indp
	$\Lambda \leftarrow \frac{p_1}{s_q}$\;
	$\hat{\ell}_{i}(1) \leftarrow 0$ for all $i$\;
\Indm
\While{job $j$ arrives} {
	$\mu_i(j) \leftarrow s_i(2\cdot \Lambda - \hat{\ell}_i(j-1))$\;
	Sort the values $\mu_i(j)$ in ascending order so that $\mu_{i_1}(j) \leq \mu_{i_2}(j) \leq \cdots \leq \mu_{i_m}(j)$ (breaking
ties arbitrarily)\;
	$A \leftarrow [m]$, $B = \emptyset$\;
	\While{$A \neq \emptyset$} {
		$s \leftarrow \min\{s_i : i \in A\}$\;
		$w \leftarrow \max\{a : s_{i_a} = s\}$\;
		$B \leftarrow B \cup \{i_w\}$\;
		$A \leftarrow A \setminus \{i_a : a \leq w\}$\;
	}
	Index the elements in $B$ such that $t_1 < t_2 < \cdots < t_{|B|}$, where $t_1$ is the smallest element, $t_2$ is the second smallest, etc.\;
	$\pi_{r_{t_1}(j)j} = 0$\;
	\For{$b=2$ to $|B|$} {
		$\pi_{r_{t_b}(j)j} = \ell_{r_{t_{b-1}}(j)}(j-1) - \ell_{r_{t_b}(j)}(j-1) +
		\left(1 - \frac{s_{t_{b-1}}}{s_{t_b}}\right)((2+\epsilon)\Lambda - \hat{\ell}_{t_{b-1}}(j-1)) + \pi_{r_{t_{b-1}}(j)j}$\;
	}
 	\If{$s_m \neq s_{t_{|B|}}$} {
		$\pi_{r_m(j)j} = \ell_{r_{t_{|B|}}(j)}(j-1) - \ell_{r_m(j)}(j-1) +
		\left(1 - \frac{s_{t_{|B|}}}{s_{m}}\right)((2+\epsilon)\Lambda - \hat{\ell}_{t_{|B|}(j)}(j-1)) + \pi_{r_{t_{|B|}}(j)j}$\;
	}
	Set all other prices to $\infty$\;
	After agent $j$ chooses a machine $q$: \;
	\Indp
		\If{$s_q = s_m$ \upshape{and} $S = \{i : \hat{\ell}_i(j-1) + \frac{p_j}{s_i} \leq 2 \cdot \Lambda\} = \emptyset$} {
			$\Lambda \leftarrow \max\{2, 2^{\left\lceil \log_2 \frac{p_j}{s_m \Lambda}\right\rceil}\} \cdot \Lambda$ \;
			$\hat{\ell}_i(j) \leftarrow 0$ for all $i$\;
		}
		\Else {
			$\hat{\ell}_q(j) \leftarrow \hat{\ell}_q(j-1) + \frac{p_j}{s_q}$\;
			$\hat{\ell}_i(j) \leftarrow \hat{\ell}_i(j-1)$ for all $i \neq q$\;
		}
	\Indm	
}
\caption{Dynamic-Related: A Dynamic Pricing Scheme for Related Machines\label{alg:dyrel}}
\end{algorithm}

We now give our dynamic pricing scheme, Dynamic-Related, which mimics the behavior of Flex-Fit. That is, the dynamic prices set
by Dynamic-Related (prior to the arrival of each job) have the following property.  Any incoming rational job $j$ will choose some machine $k$
such that Flex-Fit is free to assign job $j$ to machine $k$.

 As before,
assume that machines are sorted such that $s_1 \leq s_2 \leq \cdots \leq s_m$.  We now give some intuition behind the dynamic pricing algorithm
and discuss some of its properties.  Initially, the algorithm sets prices so that the first job is incentivized to choose machine $m$
(i.e., a machine with the fastest speed), and then obtains an initial estimate $\Lambda$ (in a manner similar to Flex-Fit).
The algorithm then enters an outer while loop which is responsible for setting prices before each next job $j$ arrives.
It begins by sorting the values $\mu_i(j) = s_i(2\cdot \Lambda - \hat{\ell}_i(j-1))$ so that $\mu_{i_1}(j) \leq \cdots \leq \mu_{i_m}(j)$.
This is useful since a job is feasible on machine $i \Leftrightarrow \hat{\ell}_i(j-1) + \frac{p_j}{s_i} \leq 2 \cdot \Lambda \Leftrightarrow p_j \leq \mu_i(j)$.

The inner while loop in Dynamic-Related is responsible for obtaining a carefully selected subsequence of machines
with strictly increasing speeds.  In particular, the inner while loop constructs a set $B$ with the property that
for all $1 \leq b \leq |B|$, we have $s_{t_1} < s_{t_2} < \cdots < s_{t_{|B|}}$ (after renaming the machine indices
in $B$ by $t_1 < t_2 < \cdots < t_{|B|}$).
This property holds due to the process by which the set $B$ is created.  Note that machine $t_1$ has the same speed
as the slowest machine (i.e., $s_1 = s_{t_1}$), since the first machine added to set $B$ is the machine in the rightmost position
in the sorted ordering $\mu_{i_1}(j) \leq \cdots \leq \mu_{i_m}(j)$ satisfying the property that it has the same speed as the slowest machine (i.e.,
machine $1$'s speed).  The inner while loop then removes all machines that appear earlier in the sorted ordering than machine $t_1$ from set $A$
(including machine $t_1$ itself).  This implies that all machines of speed $s_{t_1}$  (i.e., the slowest machine speed)
are removed from $A$, and hence the slowest machine remaining in set $A$ must have strictly larger speed.  Therefore,
when adding the second machine to set $B$, namely machine $t_2$, we get that the speed of $t_2$ satisfies $s_{t_2} > s_{t_1}$.
Repeating this process, we get that $s_{t_1} < s_{t_2} < \cdots < s_{t_{|B|}}$.
In addition, we have the property that $\mu_{t_1}(j) \leq \cdots \leq \mu_{t_{|B|}}(j)$, since each time we add a
machine to set $B$ that appears later in the sorted order $\mu_{i_1}(j) \leq \mu_{i_2}(j) \leq \cdots \leq \mu_{i_m}(j)$.

The for loop is essentially responsible for determining which machines are given finite prices (in addition to computing
the actual prices).  Note that the algorithm only sets finite prices on representatives of machines, namely the representatives of
machines $t_1,\ldots,t_{|B|}$ (possibly in addition to the representative
of machine $m$).  Hence, jobs can only be assigned to the representatives of machines $t_1,\ldots,t_{|B|}$,
and possibly the representative of machine $m$.  Dynamic-Related then updates $\Lambda$ as necessary
(in a manner similar to Flex-Fit).

Our goal now is to prove  that the prices determined by Dynamic-Related are such that for any sequence of rational
selfish jobs, the jobs will choose machines that are consistent with the choices available to Flex-Fit. This is the
same as saying that the competitive ratio of the Dynamic-Related dynamic pricing scheme is the same as the competitive
ratio of the Flex-Fit algorithm, which we know is $O(1)$ by Theorem \ref{thm:ff}. To summarize, we now seek to prove
the following theorem:

\begin{theorem}\label{thm:rel}
Dynamic-Related is an $O(1)$-competitive dynamic pricing scheme for the makespan minimization problem on related machines.
\end{theorem}

We use $FF$ and $DR$ as shorthand when referring to Flex-Fit and Dynamic-Related, respectively.
The plan for the proof is to argue that the behavior of $DR$ is consistent with the behavior of
$FF$.  Note that, at certain points, $FF$ is free to assign a job to one of  multiple
machines. Additionally, $FF$ may be free to start a new phase (or not). Hence, by the phrase ``$DR$ {\sl behaves consistently with}
$FF$,'' we mean that $DR$ sets prices so that any rational job to arrive will choose one of the permissible actions that $FF$ is free to take.

Throughout the proof, we make use of notation present in both algorithms.  In addition, when referring to the representative of a machine
$i$ when job $j$ arrives, we write $r_i$ instead of $r_i(j)$ and $\mu_i$ instead of $\mu_i(j)$ for ease of notation (we always refer to the current job $j$).
Recall that $T = \{i : \hat{\ell}_i(j-1) + \frac{p_j}{s_i} \leq (2+\epsilon)\Lambda\}$,
$S = \{i : \hat{\ell}_i(j-1) + \frac{p_j}{s_i} \leq 2 \cdot \Lambda\}$, and $k$ is the minimum
machine index in $S$ ($k = m$ if $S = \emptyset$).

Theorem~\ref{thm:rel} follows from the following four lemmas, the first of which explores some properties
of the prices $DR$ assigns, while the remaining three argue that $DR$ behaves
consistently with $FF$ in three disjoint, exhaustive cases.
The proofs of the following lemmas are deferred to Appendix~\ref{app:proofs-related}.

\begin{lemma}\label{lem:prices}
If $|B| \geq 2$, then for all $1 \leq b \leq |B| - 1$, we have $c_{r_{t_b}(j)j} \leq c_{r_{t_{b+1}}(j)j}$
if and only if $t_b \in T$.  Similarly, if $s_{t_{|B|}} \neq s_m$, then $c_{r_{t_{|B|}}(j)j} \leq c_{r_m(j)j}$
if and only if $t_{|B|} \in T$.
\end{lemma}

We now argue that $DR$ behaves consistently with $FF$ for every fixed job $j$.  Clearly, for $j=1$,
both $FF$ and $DR$ behave in the same manner (including updating $\Lambda$).  Now, consider
any job $j > 1$.  We split the proof of up into three disjoint and exhaustive cases, classified as follows:
set $T$ is empty (Lemma~\ref{lem:Tempty}), set $T$ is nonempty but set $S$ is empty (Lemma~\ref{lem:TnemptySempty}),
and finally set $S$ is nonempty (Lemma~\ref{lem:Snempty}).
Note that we always have $S \subseteq T$, so it is not possible for $T$ to be empty while $S$ is nonempty.

\begin{lemma}\label{lem:Tempty}
If $T = \emptyset$, then $DR$ assigns prices so that a rational job $j$ chooses a machine of speed $s_m$. $DR$ also updates
$\Lambda$, and resets virtual loads to zero (precisely in the same manner as $FF$).
\end{lemma}

\begin{lemma}\label{lem:TnemptySempty}
If $T \neq \emptyset$ and $S = \emptyset$, then $DR$ processes job $j$ in one of two ways.
$DR$ either (a) Sets prices so that a rational job $j$ chooses a machine of speed $s_m$, updates $\Lambda$, and resets virtual loads to zero,
or (b) $DR$ sets prices so that a rational job $j$ chooses some representative $r_i$ where $i \in T$.  Both options (a) and (b) are consistent with $FF$.
\end{lemma}

\begin{lemma}\label{lem:Snempty}
If $S \neq \emptyset$, then $DR$ sets prices so that $j$ chooses some representative $r_i$ where $s_i \leq s_k$ and $i \in T$
(which is consistent with $FF$).
\end{lemma}

This concludes the proof of Theorem \ref{thm:rel}.

\section{Lower Bounds for Unrelated Machines}
\label{sec:unrelated}

In this section, we first give a lower bound which shows that no deterministic dynamic pricing scheme can achieve a competitive ratio
better than $\Omega(m)$ for the unrelated machine setting.  Note that this competitive ratio can be achieved
by the online greedy algorithm that assigns each job $j$ to a machine that minimizes $\ell_i(j-1) + p_{ij}$.
Moreover, this behavior can be mimicked by a dynamic pricing scheme, simply by setting all prices $\pi_{ij} = 0$.

\begin{theorem}\label{thm:unreldetlb}
No deterministic dynamic pricing scheme can achieve a competitive ratio better than $\Omega(m)$ for the unrelated machine setting.
\end{theorem}
\begin{proof}

Let $D$ denote the dynamic pricing scheme, and let $OPT$ denote an optimal solution (i.e., a solution that
minimizes the makespan).  Our lower bound consists of some number of phases $k$, where $k$ can be arbitrarily large.
At the end of all $k$ phases, we argue that the makespan of the dynamic pricing scheme is at least $k \cdot m$, while $OPT$
only has a makespan of at most $(1+2\epsilon)\cdot k$ for an arbitrarily small $\epsilon > 0$.  If we can maintain this property,
this would imply a lower bound of $\Omega(m)$ on the competitive ratio.

The input consists of the following sequence of jobs.  Suppose the adversary has already introduced $j-1$ jobs,
and the dynamic pricing scheme $D$ has processed these jobs.  Hence, the dynamic pricing scheme must now set prices
$\pi_{ij}$ on each machine $i$ (before job $j$ arrives).  Then, we introduce job $j$ according to the following
two cases (depending on how the dynamic pricing scheme $D$ behaves):
\begin{enumerate}
\item If there exist two distinct machines $i,i'$ such that $\ell_i(j-1) + \pi_{ij} + \epsilon < \ell_{i'}(j-1) + \pi_{i'j}$,
introduce the job $j$ with $p_{ij} = \epsilon$, $p_{i'j} = 0$, and $p_{bj} = \infty$ for all $b \neq i,i'$.
\item Otherwise, introduce the job $j$ with $p_{1j} = 1$ and $p_{ij} = 1+2\epsilon$ for all $i > 1$.
\end{enumerate}
Note that, if in the second case all prices are $\infty$, then we assume job $j$ (i.e., agent $j$) breaks ties
by preferring machine $1$.  Each phase consists of $m$ occurrences of the second case.

The first claim we argue is that there exists an assignment of jobs to machines such that the load on each machine
grows by at most $1+2\epsilon$ at the end of each phase.  This implies that, after $k$ phases have concluded, there is an optimal solution
with a makespan of at most $(1+2\epsilon)\cdot k$.  We construct such an assignment as follows.  Observe that each time $D$
assigns prices to machines such that we fall into the first case, we simply assign the constructed job $j$
to machine $i'$ (i.e., the machine satisfying $p_{i'j} = 0$).  Hence, each time case $1$ occurs, the load of each machine
remains the same.  On the other hand, consider the jobs introduced due to the second case during a phase, of which there are $m$.
We simply assign these $m$ jobs to $m$ distinct machines.  Each such job causes the load on the machine to which it is assigned
to increase by at most $1+2\epsilon$.  Hence, there is an assignment of jobs to machines that causes the load on every machine
to increase by at most $1+2\epsilon$ after each phase completes.

We now argue that, after each phase completes, the dynamic pricing scheme $D$ sets prices such that jobs choose machines in a way that the load on
the first machine increases by $m$.  This implies that, after $k$ phases, the load on the first machine (according to the assignment
given by $D$) is at least $k \cdot m$, and hence the makespan must be at least $k \cdot m$.  We first consider the case
where $D$ sets prices such that we fall into case $1$, and argue that the load for some machine increases by $\epsilon$.  In this case, there exist two machines $i$ and $i'$ such that
$\ell_i(j-1) + \pi_{ij} + \epsilon < \ell_{i'}(j-1) + \pi_{i'j}$.  We argue that machine $i$ attains the minimum cost $c_{ij}$ for job $j$.  This
holds since $c_{ij} = \ell_i(j-1) + p_{ij} + \pi_{ij} = \ell_i(j-1) + \epsilon + \pi_{ij} < \ell_{i'}(j-1) + \pi_{i'j} = \ell_{i'}(j-1) + p_{i'j} + \pi_{i'j}$
(recall that $p_{ij} = \epsilon$ while $p_{i'j} = 0$).  Moreover, $c_{bj} = \infty$ for all $b \neq i,i'$, since $p_{bj} = \infty$.
Hence, the dynamic pricing scheme $D$ causes job $j$ to be assigned to machine $i$, which increases the load on machine $i$ by $\epsilon$.

We now consider the case where $D$ sets prices such that we fall into case $2$, and argue that each such job always chooses machine
$1$ (causing the load on machine $1$ to increase by $1$, since $p_{1j} = 1$ for each such job $j$).
Since $m$ such jobs are introduced during each phase, this would show that the load on machine $1$ increases
by $m$ after each phase completes.  In the second case,
if $\pi_{1j} = \infty$, then all prices must be $\infty$ (since there does not exist another machine $i$ satisfying
$\ell_i(j-1) + \pi_{ij} + \epsilon < \ell_{1}(j-1) + \pi_{1j} = \infty$).  Hence, since we assumed that such a job $j$
breaks ties in favor of machine $1$, the job $j$ prefers machine $1$ over all other machines (which increases
the load on machine $1$ by $1$).  On the other hand, if $\pi_{1j}$ is finite, we argue that job $j$ is also assigned to machine $1$.
Observe that in the second case, for any two machines $i$ and $i'$, we have $\ell_i(j-1) + \pi_{ij} + \epsilon \geq \ell_{i'}(j-1) + \pi_{i'j}$.
Hence, for any machine $i > 1$, we have $c_{1j} = \ell_1(j-1) + p_{1j} + \pi_{1j} = \ell_1(j-1) + 1 + \pi_{1j} \leq \ell_i(j-1) + 1 + \epsilon + \pi_{ij}
< \ell_i(j-1) + 1 + 2\epsilon  + \pi_{ij} = c_{ij}$.  This implies that machine $1$ minimizes job $j$'s cost, which in turn implies
that the load on machine $1$ increases by $1$.

Finally, observe that we can always eventually force the dynamic pricing scheme $D$ to produce prices such that we fall into case $2$.
This holds since, each time we fall into case $1$, there is an optimal solution which satisfies the property that the load on every machine
does not change.  On the other hand, $D$ causes such jobs $j$ to be assigned to some machine for which the processing time is $\epsilon$.
This implies that the assignment determined by $D$ causes the load on some machine to increase by $\epsilon$, while in an optimal solution
the load on every machine remains the same.  This process cannot go on forever, since the makespan of the schedule produced by
$D$ can be made arbitrarily bad while the makespan of an optimal solution does not change (which would result in an arbitrarily bad
competitive ratio).
\end{proof}

We now give a lower bound of $\Omega(m)$ on the expected competitive ratio for any randomized dynamic pricing scheme,
which even holds against an oblivious adversary.

\begin{theorem}\label{thm:unrelrandlb}
No randomized dynamic pricing scheme can achieve an expected competitive ratio better than $\Omega(m)$ for the unrelated machine setting.
\end{theorem}
\begin{proof}
Let $D$ denote the randomized dynamic pricing scheme, and let $OPT$ denote an optimal solution (i.e., a solution that
minimizes the makespan).  Our lower bound draws on ideas from the proof of our deterministic lower bound (i.e., the proof of
Theorem~\ref{thm:unreldetlb}).  Let $\epsilon > 0$ be arbitrarily small.  The input sequence will consist of some number of jobs $n$,
where we make $n$ arbitrarily large.

In particular, the input sequence is iteratively constructed by the adversary in a deterministic manner.
After the first $j-1$ jobs have been determined (independent of the randomness of the algorithm), the adversary constructs
job $j$ according to the following two cases, one of which must happen with probability at least $\frac{1}{2}$
(depending on the code of the dynamic pricing scheme $D$, but independent of any random coin flips).  Here,
the probability is taken over the coin flips of the algorithm, and the probability is computed as if
the first $j-1$ jobs constructed so far are given as input to $D$ (note that the jobs are only imagined
to be given as input):
\begin{enumerate}
\item There exist two distinct machines $i,i'$ such that $\ell_i(j-1) + \pi_{ij} + \epsilon < \ell_{i'}(j-1) + \pi_{i'j}$ with probability at
least $\frac{1}{2}$.  Note that, conditioned on this event occurring, there must exist a fixed pair of machines $(i,i')$
(obtained deterministically) such that the event $\ell_i(j-1) + \pi_{ij} + \epsilon < \ell_{i'}(j-1) + \pi_{i'j}$ occurs with
probability at least $\frac{1}{m^2}$, since there are at most $m^2$ such pairs.
For this fixed pair $(i,i')$, the adversary constructs the job $j$
with $p_{ij} = \epsilon$, $p_{i'j} = 0$, and $p_{bj} = \infty$ for all $b \neq i,i'$.
\item Otherwise, introduce the job $j$ with $p_{1j} = 1$ and $p_{ij} = 1+2\epsilon$ for all $i > 1$.
\end{enumerate}
Observe that one of these two cases must happen with probability at least $\frac{1}{2}$.  Hence, the adversary
constructs job $j$ depending on which such case occurs with higher probability.
Note that, if in the second case all prices are $\infty$, then we assume job $j$ (i.e., agent $j$) breaks ties
by preferring machine $1$.

Let $n_1$ be the number of jobs constructed due to case $1$, and $n_2$ be the number of jobs constructed
due to case $2$.  Note that $n_1$ and $n_2$ are fixed (i.e., they are not random variables), and moreover
we have $n_1 + n_2 = n$.  Using similar reasoning as in the proof of Theorem~\ref{thm:unreldetlb}, $OPT \leq (1+2\epsilon) \cdot \lceil\frac{n_2}{m}\rceil$,
since we can always assign case $1$ jobs to the machine on which a load of $0$ is incurred.  Moreover, for every $m$
jobs constructed due to case $2$, we can always put the $m$ jobs on $m$ distinct machines, increasing the makespan
by at most $(1+2\epsilon)$.

On the other hand, we argue that the expected makespan of the dynamic pricing scheme $D$ must be large.
Observe that for each job $j$ constructed due to case $1$, the dynamic pricing scheme's expected sum of loads
increases as follows: $\mathbb{E}[\sum_{i}\ell_i(j) - \ell_i(j-1)] \geq \epsilon \cdot \frac{1}{2m^2}$.  This holds since,
for each such job $j$, if the event $\ell_k(j-1) + \pi_{kj} + \epsilon < \ell_{k'}(j-1) + \pi_{k'j}$ occurs for the fixed
pair of machines $(k,k')$, which happens with probability at least $\frac{1}{2m^2}$, then job $j$ induces a load of $\epsilon$
on machine $k$.  Hence, we have $\mathbb{E}[\sum_i \ell_i(n)] \geq n_1 \cdot \epsilon \cdot \frac{1}{2m^2}$.
Moreover, for each job $j$ constructed due to case $2$, the load on machine $1$ increases by $1$ with probability at least $\frac{1}{2}$
(since if the event corresponding to case $2$ occurs, which happens with probability at least $\frac{1}{2}$, job $j$
is assigned to machine $1$).  Hence, we have $\mathbb{E}[\ell_1(n)] \geq n_2 \cdot \frac{1}{2}$.  Putting it all together,
we get
$$\mathbb{E}\left[\max_i \ell_i(n)\right] \geq \mathbb{E}\left[\max\left\{\frac{1}{m} \sum_i \ell_i(n), \ell_1(n)\right\}\right] \geq
\max\left\{\frac{\epsilon \cdot n_1}{2m^3},\frac{n_2}{2}\right\}.$$

We consider the following two cases.  In the first case, we have $n_2 \leq t \cdot m$ for some large integer $t$, and in the
other case we have $n_2 > t \cdot m$.  Note that we assume $n_2 > 0$, as otherwise $OPT = 0$ and the competitive ratio is arbitrarily bad.
In the first case, the expected competitive ratio is at least
$\frac{\epsilon \cdot (n - n_2)}{2m^3 \cdot (1+2\epsilon) \cdot \lceil\frac{n_2}{m}\rceil} \geq \frac{\epsilon \cdot (n - t \cdot m)}{2m^3 \cdot (1+2\epsilon)\cdot(t + 1)} = \Omega(\sqrt{n})$,
which can be made arbitrarily bad (note that we choose $n$ sufficiently large so that $\sqrt{n} \gg t \cdot m^3$).
In the second case, the expected competitive ratio is at least
$\frac{n_2}{2(1+2\epsilon) \cdot \lceil\frac{n_2}{m}\rceil} = \Omega(m)$.  Note that, in the second case,
$OPT$ is large and hence our lower bound rules out additive constants in the competitive ratio of the dynamic pricing scheme.

\end{proof}

\section{Greedy and Static Pricing are Equivalent}
\label{sec:static_pricing}

In this section, we relate the greedy algorithm and static pricing schemes.
Our main theorem, Theorem \ref{thm:equiv}, is given below. The proof of Theorem \ref{thm:equiv} appears in Appendix \ref{apx:lowerequiv}.
\begin{theorem}\label{thm:equiv}
If the greedy algorithm that assigns each job $j$ to a machine $i$ that minimizes $\ell_i(j-1) + p_{ij}$
is $c$-competitive in some machine model (either identical, related, restricted,
or unrelated), then the static pricing scheme that sets all prices to zero is also $c$-competitive.
Moreover, a lower bound of $c$ on the (expected) competitive ratio of the greedy algorithm implies a lower
bound of $\Omega(c)$ for any deterministic or randomized static pricing scheme (for all machine models).
The randomized lower bound holds as long as the greedy lower bound does not specify how ties are broken.
\end{theorem}


\bibliographystyle{ACM-Reference-Format}
\bibliography{lbrefs}

\appendix

\section{Static Pricing Schemes $\equiv$ Greedy}\label{apx:lowerequiv}

Recall that a static pricing scheme is completely determined by an $m$-dimensional vector $\pi = (\pi_{1*},\ldots,\pi_{m*})$,
which is set in advance before any agents arrive.  First, we claim that static pricing schemes can always mimic
the online greedy algorithm where the greedy choice corresponds to placing the incoming job $j$ on the machine $i$
that minimizes $\ell_i(j-1) + p_{ij}$.  In particular, this online greedy algorithm can be mimicked by a
static pricing scheme, simply by setting all prices to $0$ (i.e., $\pi = (0,\ldots,0)$).  With these
static prices, all incoming agents $j$ choose a machine $i$ that minimizes
$\ell_i(j-1) + p_{ij} + \pi_{i*} = \ell_i(j-1) + p_{ij}$ (since $\pi_{i*} = 0$ for all $i$).  Hence, agents
choose a machine in a manner that is consistent with the greedy algorithm's choices.

Theorem \ref{thm:equiv} immediately follows from the discussion above, along with Lemma~\ref{lem:detequiv}
(our deterministic lower bound) and Lemma~\ref{lem:randequiv} (our randomized lower bound) below.
We first give our deterministic lower bound in the following lemma.

\begin{lemma}\label{lem:detequiv}
A lower bound of $c$ on the competitive ratio of the greedy algorithm that assigns each job $j$ to a machine $i$
that minimizes $\ell_i(j-1) + p_{ij}$ implies a lower bound of $\Omega(c)$ on the competitive ratio of
any deterministic static pricing scheme (for the identical, related, restricted, and unrelated machine models).
\end{lemma}
\begin{proof}
Assume we have some static pricing scheme, the prices of which are given by $\pi = (\pi_{1*},\ldots,\pi_{m*})$.
We denote by $\pi_{\max}$ the largest price determined by the static pricing scheme, so that $\pi_{\max} = \max_i \pi_{i*}$.
For any input sequence $\sigma$, we denote by $ALG(\sigma)$ the makespan of the greedy algorithm, and let $ALG'(\sigma)$
denote the makespan achieved by the static pricing scheme.
Suppose there exists some adversarial sequence of $n$ input jobs which gives witness to the lower bound
of $c$ on the competitive ratio of the greedy algorithm.  That is, for any additive constant $a$, there exists
an input sequence $\sigma$ which satisfies the property that
$ALG(\sigma) > c \cdot L^*(\sigma) + a$ (recall that $a$ represents the additive constant in the competitive ratio,
and $L^*(\sigma)$ denotes the makespan of an optimal solution on input $\sigma$).  We note that $L^*(\sigma)$ can
be made arbitrarily large by scaling all jobs appropriately.

Our goal is to modify the input sequence $\sigma$ by prepending a few jobs at the beginning of the sequence,
obtaining a new input sequence $\sigma'$ which yields a comparable guarantee on the competitive ratio for any
static pricing scheme.  In particular, our aim is to introduce $m$ input jobs with the property that the loads
plus prices on all machines are flattened out.  More formally, we introduce jobs $1,\ldots,m$ with the property
that $\ell_i(m) + \pi_{i*} = \ell_{i'}(m) + \pi_{i'*} = \pi_{\max}$ for all machines $i,i'$.  Once this is done, we can
then introduce jobs in the sequence $\sigma$ to the static pricing scheme, and agents will choose machines
in precisely the same manner as jobs are placed by the online greedy algorithm.  Note that, initially,
all loads are $0$ (i.e., $\ell_i(0) = 0$ for all $i$).

We first describe the process in the unrelated machine setting.  Here, for each machine $i$,
we introduce the job $j$ such that $p_{ij} = \pi_{\max} - \pi_{i*}$ and $p_{i'j} = \infty$ for all $i' \neq i$.
These jobs may be introduced in an arbitrary order (notice that there are $m$ such jobs).  At the end of this sequence of $m$ jobs, since each
machine $i$ gets exactly one of the $m$ jobs, we have $\ell_i(m) = \pi_{\max} - \pi_{i*}$, and hence
$\ell_i(m) + \pi_{i*} = \pi_{\max}$ for all machines $i$.

We now turn our attention to the identical machine setting.
In this model, we must introduce jobs in a specific order.  In particular, we sort machines in increasing order
of their prices, so that we have $\pi_{i_1 *} \leq \pi_{i_2 *} \leq \cdots \leq \pi_{i_m *}$.  We introduce the following jobs:
for each $j = 1,\ldots,m$, we set $p_j = \pi_{\max} - \pi_{i_j*}$.  Notice that the first job chooses machine $i_1$, since
$\ell_{i_1}(0) + p_1 + \pi_{i_1 *} \leq \ell_i(0) + p_1 + \pi_{i*}$ for all machines $i$ (if there are ties, then the agent
can choose a machine arbitrarily and everything still goes through).  Moreover, we argue that every job $2 \leq j \leq m$ chooses
machine $i_j$.  This holds since, for all $k < j$, the cost on machine $i_k$ is given by
$\ell_{i_k}(j-1) + p_j + \pi_{i_k*} = \pi_{\max} - \pi_{i_k *} + p_j + \pi_{i_k*} = \pi_{\max} + p_j$, while the cost on machine $i_j$
is $\ell_{i_j}(j-1) + p_j + \pi_{i_j*} = (\pi_{\max} - \pi_{i_j*}) + \pi_{i_j*} = \pi_{\max}$.  Moreover, for all machines $k \geq j$,
the cost on machine $i_k$ is given by $\ell_{i_k}(j-1) + p_j + \pi_{i_k*} = p_j + \pi_{i_k*} \geq p_j + \pi_{i_j*}$, which is the
cost on machine $i_j$.  Hence, once all $m$ jobs have been introduced, we have the property that $\ell_i(m) + \pi_{i*} = \pi_{\max}$ for all
machines $i$.

Finally, we describe the process by which we compute $p_j$ in the related machine setting for each job
$1 \leq j \leq m$.  In a continuous manner, starting from $0$, we continuously increase $p_j$ until the machine $i$ that minimizes
$\ell_i(j-1) + \frac{p_j}{s_i} + \pi_{i*}$ (note that the minimizing machine may change) satisfies the property that
$\ell_i(j-1) + \frac{p_j}{s_i} + \pi_{i*}$ equals $\pi_{\max}$ (at which point the process stops). Once equality is attained,
we simply introduce the job $j$ with this value of $p_j$.  Again, it is easy to see that, after all $m$ jobs have been introduced,
we have the property that $\ell_i(m) + \pi_{i*} = \pi_{\max}$ for all machines $i$.

As mentioned, to obtain a lower bound on the competitive ratio of $ALG'$, we first introduce the $m$ jobs as mentioned above
(i.e., flatten things out), followed by the input sequence $\sigma$ (yielding the sequence $\sigma'$).
Observe that $ALG'(\sigma') \geq ALG(\sigma)$,
since the static pricing scheme incurs the same load on every machine as the online greedy algorithm, in addition to the load
of one of the $m$ jobs introduced at the beginning of the entire sequence.  By assumption, we have
$ALG(\sigma) > c \cdot L^*(\sigma) + a$.  Finally, we have $L^*(\sigma') \leq L^*(\sigma) + \pi_{\max}$,
since one feasible solution is to assign jobs to machines precisely as an optimal solution on input $\sigma$ does,
along with placing the $m$ initial jobs on $m$ distinct machines (incurring an additional load of at most $\pi_{\max}$).
Putting everything together, we get
$ALG'(\sigma') \geq ALG(\sigma) > c \cdot L^*(\sigma) + a \geq c \cdot (L^*(\sigma') - \pi_{\max}) + a$.  Since
$L^*(\sigma') \geq L^*(\sigma)$ can be made arbitrarily large (in particular we can make $L^*(\sigma) \gg \pi_{\max}$),
we have $ALG'(\sigma') > \Omega(c) \cdot L^*(\sigma') + a$, giving the lemma.
\end{proof}

We now provide our randomized lower bound.
\begin{lemma}\label{lem:randequiv}
A lower bound of $c$ on the competitive ratio of the greedy algorithm that assigns each job $j$ to a machine $i$
that minimizes $\ell_i(j-1) + p_{ij}$ implies a lower bound of $c$ on the expected competitive ratio of
any randomized static pricing scheme (for the identical, related, restricted, and unrelated machine models).
This implication holds as long as the greedy lower bound does not specify how ties are broken.
\end{lemma}
\begin{proof}
Assume we have some randomized static pricing scheme, the prices of which are given by $\pi = (\pi_{1*},\ldots,\pi_{m*})$,
where each $\pi_{i*}$ for $1 \leq i \leq m$ is a random variable.
The idea behind the proof is to scale each job's processing times up by such a large amount that the randomly produced
prices on all machines become negligible.  In particular, we scale jobs up by such a large amount
that even the largest price is negligible compared to the smallest (non-zero) processing time of any
job on any machine.  We note that the following argument holds for all machine models, and in particular for
identical, related, restricted, and unrelated machines.

We scale jobs up as follows.  Consider feeding the input sequence $\sigma$ as input to the greedy
algorithm, and suppose that each time a job $j$ is assigned to a machine, there are no
ties to be broken (i.e., $\ell_i(j-1) + p_{ij} \neq \ell_k(j-1) + p_{kj}$ for all machines $i \neq k$).
This is not true in general (in particular, it is never true for the identical machines setting, as even
the first job faces the same load on all machines, namely zero), but we consider this case first for simplicity.
Over the run of the greedy algorithm on input $\sigma$, consider the smallest gap that ever exists between a pair of machines.
Namely, consider $\delta = \min_j \min_{k \neq i} |\ell_i(j-1) + p_{ij} - \ell_k(j-1) - p_{kj}|$ (note that $\delta \neq 0$
by our assumption).  As long as we scale up jobs so that the random prices are $\ll \delta$, then jobs
are assigned by the randomized static pricing scheme in precisely the same manner as they are by the greedy
algorithm.  This holds since, if in the greedy algorithm machine $i$ minimizes $\ell_i(j-1) + p_{ij}$
(note that this machine is unique by our assumption), then machine $i$ also minimizes $\ell_i(j-1) + p_{ij} + \pi_{i*}$.
In particular, for all $k \neq i$ we have $\ell_k(j-1) + p_{kj} + \pi_{k*} \geq \ell_i(j-1) + p_{ij} + \delta + \pi_{k*} > \ell_i(j-1) + p_{ij} + \pi_{i*}$.

We now discuss the scenario when some jobs need to break ties among at least two machines over
the run of the greedy algorithm on input $\sigma$.  In this case, we define $\delta$ to be the smallest non-zero
gap $\delta$ that ever exists over all possible runs of the greedy algorithm on input $\sigma$.  By all possible runs,
we mean considering all possible ways that the greedy algorithm can resolve ties among machines for each job.
Similarly in this case, we scale all jobs so that each job's processing times are significantly larger relative to the prices.
Now, for a job $j$, if all possible runs result in a unique machine $i$ that minimizes $\ell_i(j-1) + p_{ij}$,
then the randomized static pricing scheme will assign it to the same machine.  On the other hand, if there exists
a run that results in job $j$ facing ties among multiple machines, then the randomized static pricing scheme
assigns $j$ to one such machine (note that, by the assumption in the statement of the lemma, job $j$ is free
to be assigned to any such machine since the greedy lower bound does not specify how ties are broken).
\end{proof}

\section{An illustrative Example}\label{app:example}


We motivate
why dynamic pricing is useful via a small example.  We do so by comparing the schedule produced without any
pricing to the schedule produced via a dynamic pricing scheme.
For ease of presentation, we assume that our scheme knows the value of the optimal makespan, which we denote by $L^*$.
Schedules obtained without pricing are equivalent to schedules produced by the greedy algorithm that assigns each job $j$ to a machine $i$ that minimizes $\ell_i(j-1) + \frac{p_j}{s_i}$.

\begin{figure}
\includegraphics[width=\textwidth]{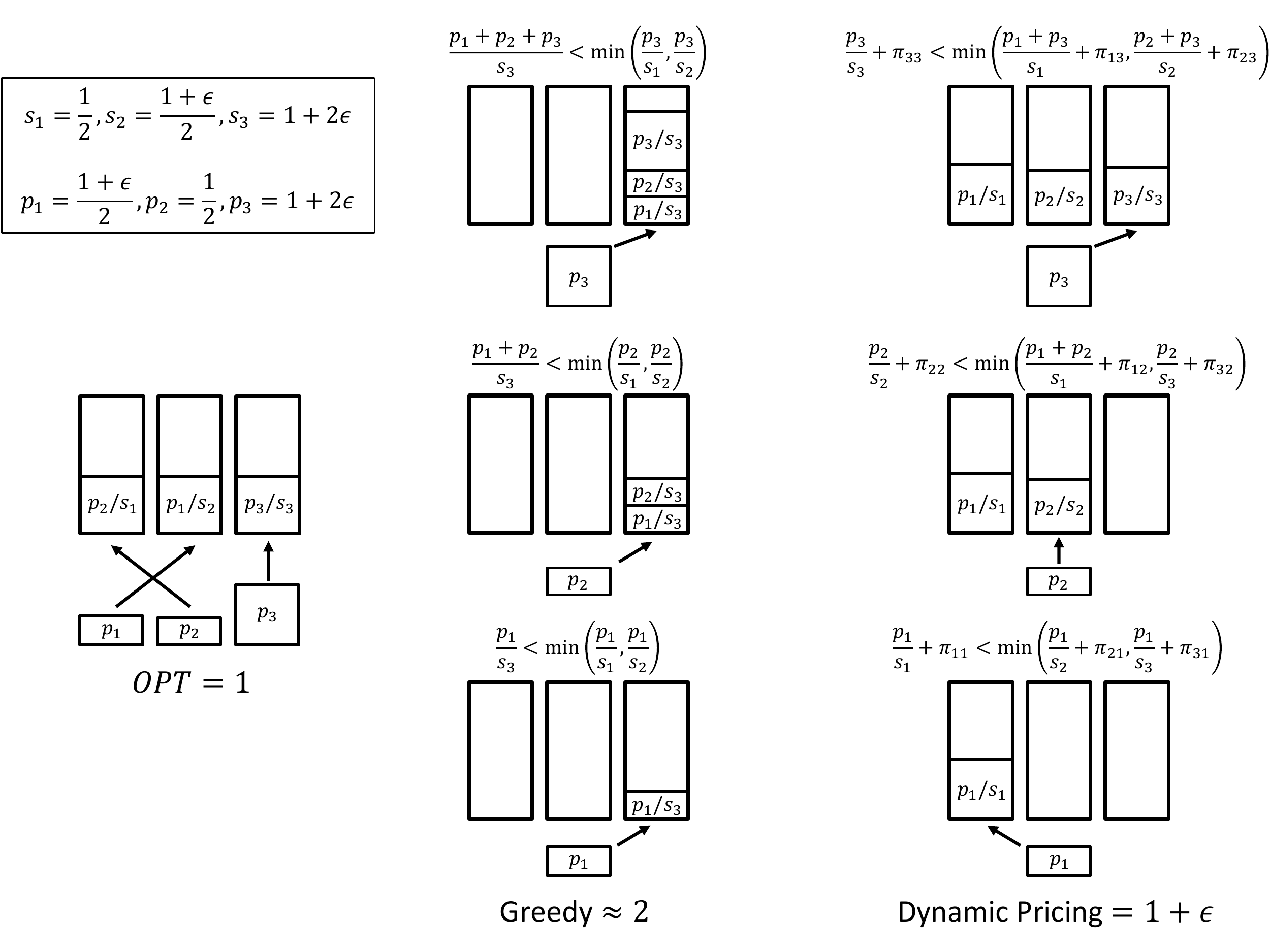}
\caption{Example of the optimal makespan, the greedy algorithm, and dynamic pricing set using Algorithm \ref{alg:dyrel}.}
\center{\begin{tabular}{lll}
  \toprule
  \multicolumn{3}{c}{Prices}\\
  \midrule
   $\pi_{11}=0$ & $\pi_{21}>\epsilon$ & $\pi_{31}\approx 1+ \pi_{21}$\\
   $\pi_{12}=0$ & $\pi_{22}<1+\epsilon+\epsilon/(1+\epsilon)$ & $\pi_{32}\approx 2$ \\
   $\pi_{13}=0$ & $\pi_{23}\approx \epsilon$ & $\pi_{33}\approx 3/2$ \\
  \bottomrule
\end{tabular}}
\label{fig:example}
\end{figure}

In our example (given as Figure \ref{fig:example}) there are $m=3$ machines, with speeds $s_1 = \frac{1}{2}$, $s_2 = \frac{1}{2}(1+\epsilon)$, and $s_3 = 1+2\epsilon$;  and $n=3$ jobs, with sizes $p_1 = \frac{1}{2}(1+\epsilon)$, $p_2 = \frac{1}{2}$, and $p_3 = 1+2\epsilon$.
The left, middle, and right columns show the optimal assignment, the greedy assignment, and the assignment obtained by our dynamic pricing scheme, respectively. In the middle and right columns, the arrival order is from bottom to top.

{\bf Optimal makespan:} The optimal makespan is $L^* = 1$, achieved by assigning job $1$ to machine $2$, job $2$ to machine $1$, and job $3$ to machine $3$.

{\bf Greedy:} The greedy algorithm assigns job $1$ to machine $3$, since the machine $i$ that minimizes $\frac{p_1}{s_i}$ is the fastest machine (initially, all loads are $0$).  Job $2$ is also assigned to machine $3$, since
$\ell_3(1) + \frac{p_2}{s_3} = \frac{2+\epsilon}{2(1+2\epsilon)} < \frac{1}{1+\epsilon} = \frac{p_2}{s_2} < \frac{p_2}{s_1}$ (for
sufficiently small $\epsilon > 0$).  Lastly, job $3$ is also assigned to machine $3$, since
$\ell_3(2) + \frac{p_3}{s_3} = \frac{2 + \frac{5\epsilon}{2}}{1+2\epsilon} < \frac{2(1+2\epsilon)}{1+\epsilon} = \frac{p_3}{s_2} < \frac{p_3}{s_1}$.
Hence, the greedy algorithm assigns all jobs to machine $3$, resulting in a makespan of $\frac{p_1 + p_2 + p_3}{s_3} \approx 2$.

{\bf Pricing scheme:}
Our dynamic pricing scheme sets prices before the arrival of each job, and are independent of the type of the incoming job.
The prices we use are the prices generated by our $O(1)$-competitive dynamic pricing scheme.
We defer the explanation of how to construct these prices to Section \ref{sec:prices-related}.
The following are the prices set prior to the arrival of job $1$:
$\pi_{11} = 0$, $\pi_{21} = \left(1 - \frac{s_1}{s_2}\right)\cdot \left(2+\frac{\epsilon}{2}\right)L^* =
\left(2+\frac{\epsilon}{2}\right)\cdot \frac{\epsilon}{1+\epsilon} > \epsilon$, and
$\pi_{31} = \left(1 - \frac{s_2}{s_3}\right)\cdot \left(2+\frac{\epsilon}{2}\right)L^* + \pi_{21} \approx 1 + \pi_{21}$.
Hence, job $1$ chooses machine $1$,
since $c_{11} = \frac{p_1}{s_1} + \pi_{11} = 1+\epsilon < 1 + \pi_{21} = \frac{p_1}{s_2} + \pi_{21} = c_{21} <
\frac{1}{2} + 1 + \pi_{21} \approx \frac{p_1}{s_3} + \pi_{31} = c_{31}$.

Prior to the arrival of job $2$, the dynamic pricing scheme sets prices as follows:
$\pi_{12} = 0, \pi_{22} = \ell_1(1) + \left(1 - \frac{s_1}{s_2}\right)\left(\left(2+\frac{\epsilon}{2}\right)L^* - \ell_1(1)\right) < 1+\epsilon + \frac{\epsilon}{1+\epsilon}$,
and $\pi_{32} = \left(1 - \frac{s_2}{s_3}\right)\left(2 + \frac{\epsilon}{2}\right)L^* + \pi_{22} \approx 2$.  Hence,
job $2$ chooses machine $2$, since $c_{22} = \frac{p_2}{s_2} + \pi_{22} < \frac{2p_2}{1+\epsilon} + 1+\epsilon + \frac{\epsilon}{1+\epsilon}
= \ell_1(1) + \frac{p_2}{s_1} = c_{12} < \frac{p_2}{1+2\epsilon} + 2 \approx \frac{p_2}{s_3} + \pi_{32} = c_{32}$.

Finally, prior to the arrival of job $3$, the dynamic pricing scheme sets the following prices:
$\pi_{13} = 0$, $\pi_{23} = \ell_1(2) - \ell_2(2) + \left(1 - \frac{s_1}{s_2}\right)\left(\left(2 + \frac{\epsilon}{2}\right)L^* - \ell_1(2)\right) \approx \epsilon$,
and $\pi_{33} = \ell_2(2) + \left(1 - \frac{s_2}{s_3}\right)\left(\left(2+\frac{\epsilon}{2}\right)L^* - \ell_2(2)\right) + \pi_{23} \approx \frac{3}{2}$.
Hence, job $3$ chooses machine $3$, since $c_{33} = \frac{p_3}{s_3} + \pi_{33} \approx \frac{5}{2}$, while
$c_{13} = \ell_1(2) + \frac{p_3}{s_1} = 1+\epsilon + 2p_3 \approx 3$ and
$c_{23} = \ell_2(2) + \frac{p_3}{s_2} + \pi_{23} = \frac{1}{1+\epsilon} + \frac{2p_3}{1+\epsilon} + \pi_{23} \approx 3$.

Since machine $1$ has the highest load, the schedule produced by our dynamic pricing scheme achieves a makespan of $\ell_1(3) = 1+\epsilon$.
This example can be extended to show that greedy can be as bad as $\Omega(\log m)$-competitive, while in
contrast our dynamic pricing scheme is $O(1)$-competitive.

\medskip

\section{Missing proofs from Section~\lowercase{\ref{sec:prices-related}}}
\label{app:proofs-related}

{\bf Proof of Lemma~\ref{lem:phaseload}:}
\begin{proof}
Fix any arriving job $j$, and assume that $S \neq \emptyset$.  Since $S \neq \emptyset$, we must also have that $T \neq \emptyset$.
Hence, the algorithm assigns the job to any representative of a machine $i$ when $j$ arrives, where $s_i \leq s_k$ and
$i \in T$ (recall that $k$ is the minimum index in $S$).  We wish to show that
$\hat{\ell}_{r_i(j)}(j-1) + \frac{p_j}{s_i} \leq (2+\epsilon)\Lambda$.  Let $i$ be any machine in $T$, which implies
that $\hat{\ell}_i(j-1) + \frac{p_j}{s_i} \leq  (2+\epsilon) \Lambda$.
The representative $r_i(j)$ must be a machine satisfying $\hat{\ell}_{r_i(j)}(j-1) \leq \hat{\ell}_{i}(j-1)$.
Hence, we have $\hat{\ell}_{r_i(j)}(j-1) + \frac{p_j}{s_i} \leq \hat{\ell}_i(j-1) + \frac{p_j}{s_i} \leq  (2+\epsilon) \Lambda$,
which gives the lemma.
\end{proof}

\noindent {\bf Proof of Lemma~\ref{lem:nonempty}:}

\begin{proof}
Suppose $\Lambda \geq L^*$, and let $j > 1$ be any arriving job after our estimate $\Lambda$ exceeds $L^*$.
Note that, once $\Lambda \geq L^*$, virtual loads are reset to zero, and hence we only consider jobs that
have arrived since then.  Assume towards a contradiction that $S = \emptyset$.
Let $f$ be the fastest machine satisfying $\hat{\ell}_f(j-1) \leq L^*$, namely $f = \max\{i : \hat{\ell}_i(j-1) \leq L^* \}$.
If such a machine does not exist, we set $f = 0$ (in fact, we will show that such a machine must exist).
For any machine $i$ such that $s_i = s_m$, we must have $\hat{\ell}_i(j-1) > L^*$, since otherwise we would have a contradiction to the fact
that $S = \emptyset$: $\hat{\ell}_i(j-1) + \frac{p_j}{s_i} = \hat{\ell}_i(j-1) + \frac{p_j}{s_m} \leq L^* + L^* \leq 2\cdot \Lambda$
(note that $\frac{p_j}{s_m} \leq L^*$).  In particular, we have $s_f < s_m$ (assuming $f \geq 1$).

Now, let $\Gamma = \{i : s_i > s_f\}$, and note that each machine in $\Gamma$ has load strictly more than $L^*$
(if $f = 0$, we let $\Gamma$ be the set of all machines).
By the fact that $s_f < s_m$, we know that $\Gamma \neq \emptyset$ (if $f=0$, this is also the case).
Let $J_i$ be the set of jobs assigned to machine $i$ after the estimate $\Lambda \geq L^*$,
and define $J_i^*$ to be the set of jobs that the optimal solution
assigns to $i$ after $\Lambda \geq L^*$.  Then we have:

$$ \sum_{i \in \Gamma}\sum_{j \in J_i}\frac{p_j}{s_m} = \frac{1}{s_m}\sum_{i \in \Gamma}s_i\sum_{j \in J_i}\frac{p_j}{s_i} >
\frac{1}{s_m}\sum_{i \in \Gamma}s_i \cdot L^* \geq \frac{1}{s_m}\sum_{i \in \Gamma}s_i \sum_{j \in J_i^*}\frac{p_j}{s_i} =
\sum_{i \in \Gamma}\sum_{j \in J_i^*}\frac{p_j}{s_m}. $$
This implies that there must exist at least one job $b \leq j-1$ such that the online algorithm assigns $b$ to a machine
$i' \in \Gamma$ while the optimal solution assigns $b$ to a machine $i^* \not\in \Gamma$.  Hence, the set $\Gamma$
is neither empty nor the entire set of machines (note this shows that $f \geq 1$, and hence there always exists
a machine $i$ satisfying $\hat{\ell}_i(j-1) \leq L^*$).

Since $b$ was assigned by the optimal solution to a machine $i^* \not\in \Gamma$, we have the property that $\frac{p_b}{s_{i^*}} \leq L^*$.
Moreover, since $i^* \not\in \Gamma$, we know that $s_{i^*} \leq s_f$.  Hence, we have
$\frac{p_b}{s_f} \leq \frac{p_b}{s_{i^*}} \leq L^*$.  In addition, since $\hat{\ell}_f(j-1) \leq L^*$,
we have $\hat{\ell}_f(b-1) + \frac{p_b}{s_f} \leq \hat{\ell}_f(j-1) + \frac{p_b}{s_f} \leq 2 \cdot L^* \leq 2 \cdot \Lambda$
(since $b \leq j-1$ and virtual loads can only grow within a phase).  Thus, when job $b$ arrived, the set $S$ was nonempty (and hence,
$T$ was also nonempty).  This implies that the online algorithm assigned job $b$ to the representative $i' = r_i(b)$ of some machine
$i \in T$ where $s_i \leq s_k$ (recall that $k$ is the minimum machine index in $S$).  In particular, machine $f \in S \subseteq T$,
and hence $s_i \leq s_k \leq s_f$.  Since $r_i(b) \in \Gamma$, we have
$s_i \leq s_f < s_{r_i(b)} = s_i$, which yields a contradiction and gives the lemma.
\end{proof}

\noindent {\bf Proof of Lemma~\ref{lem:prices}:}

\begin{proof}

First, we suppose $|B| \geq 2$. We claim that for all $1 \leq b \leq |B|-1$, we have
$c_{r_{t_b}j} \leq c_{r_{t_{b+1}}j} \Leftrightarrow \hat{\ell}_{t_b}(j-1) + \frac{p_j}{s_{t_b}} \leq (2+\epsilon)\Lambda$.
In particular, we have the following:
\begin{align*}
c_{r_{t_b}j} \leq c_{r_{t_{b+1}}j} &\Longleftrightarrow
\ell_{r_{t_b}}(j-1) + \frac{p_j}{s_{r_{t_b}}} + \pi_{r_{t_b}j} \leq \ell_{r_{t_{b+1}}}(j-1) + \frac{p_j}{s_{r_{t_{b+1}}}} + \pi_{r_{t_{b+1}}j}\\
&\Longleftrightarrow p_j\left(\frac{1}{s_{t_b}} - \frac{1}{s_{t_{b+1}}}\right) \leq \ell_{r_{t_{b+1}}}(j-1) - \ell_{r_{t_b}}(j-1) + \pi_{r_{t_{b+1}}j} - \pi_{r_{t_b}j}.
\end{align*}
Substituting for $\pi_{r_{t_{b+1}}j}$, we find that the right hand side of the expression is given by:
\begin{align*}
\ell_{r_{t_{b+1}}}(j-1) &- \ell_{r_{t_b}}(j-1) + \\
&\left[\ell_{r_{t_b}}(j-1) - \ell_{r_{t_{b+1}}}(j-1) +
\left(1 - \frac{s_{t_b}}{s_{t_{b+1}}}\right)((2+\epsilon)\Lambda - \hat{\ell}_{t_b}(j-1)) + \pi_{r_{t_b}j}\right] - \pi_{r_{t_b}j}\\
&= \left(1 - \frac{s_{t_b}}{s_{t_{b+1}}}\right)((2+\epsilon)\Lambda - \hat{\ell}_{t_b}(j-1)).
\end{align*}
Hence, we have:
\begin{align*}
c_{r_{t_b}j} \leq c_{r_{t_{b+1}}j}
&\Longleftrightarrow
p_j\left(\frac{1}{s_{t_b}} - \frac{1}{s_{t_{b+1}}}\right) \leq \left(1 - \frac{s_{t_b}}{s_{t_{b+1}}}\right)((2+\epsilon)\Lambda - \hat{\ell}_{t_b}(j-1))\\
&\Longleftrightarrow p_j \leq s_{t_b}((2+\epsilon)\Lambda - \hat{\ell}_{t_b}(j-1))
\Longleftrightarrow \hat{\ell}_{t_b}(j-1) + \frac{p_j}{s_{t_b}} \leq (2+\epsilon)\Lambda,
\end{align*}
where we used the fact that $s_{t_b} < s_{t_{b+1}}$ when dividing both sides by $(\frac{1}{s_{t_b}} - \frac{1}{s_{t_{b+1}}})$ (in the second step).
Thus, in the end we conclude that job $j$ prefers machine $r_{t_b}$ to machine $r_{t_{b+1}}$ if and only if machine $t_b$
belongs to set $T$.  Note that, due to tie-breaking issues, it is possible for a job $j$ to choose machine $r_{t_{b+1}}$ when
$p_j = s_{t_b}((2+\epsilon)\Lambda - \hat{\ell}_{t_b}(j-1))$.  We eventually argue that this
does not create any issues (note that if $p_j < s_{t_b}((2+\epsilon)\Lambda - \hat{\ell}_{t_b}(j-1))$, then
job $j$ strictly prefers machine $r_{t_b}$ to machine $r_{t_{b+1}}$).  In addition, by the same reasoning,
we can conclude that if $s_{t_{|B|}} \neq s_m$, then $c_{r_{t_{|B|}}j} \leq c_{r_mj} \Leftrightarrow t_{|B|} \in T$.
\end{proof}

\noindent {\bf Proof of Lemma~\ref{lem:Tempty}:}

\begin{proof}
We first consider the case when set $T$ is empty.  In this case, $FF$ allows job $j$ to be assigned to any machine of the fastest speed (namely, speed $s_m$),
and begins a new phase by updating $\Lambda$ and sets all virtual loads to $0$.  Since set $T$ is empty, we know that
for all machines $i$, we have $\hat{\ell}_i(j-1) + \frac{p_j}{s_i} > (2+\epsilon)\Lambda$.  If $s_{t_{|B|}} \neq s_m$,
then by Lemma~\ref{lem:prices}, job $j$ strictly prefers machine $r_m$ to the representatives of all machines in $B$, namely
$c_{r_mj} < c_{r_{t_b}j}$ for all $1 \leq b \leq |B|$ (note that these are the only machines that receive finite prices,
and hence are the only machines with a finite cost to the job).

If $s_{t_{|B|}} = s_m$, then we have two cases depending on the size of $B$.  If $|B| = 1$, then machine $r_m$ is the only
machine that receives a finite price.  If $|B| > 1$, then since $t_{|B|-1} \not\in T$, by Lemma~\ref{lem:prices} we know that job $j$
strictly prefers machine $r_{t_{|B|}}$ (which has the same speed as machine $m$) to machine $r_{t_{|B| - 1}}$ (machine $r_{t_{|B|}}$
is also strictly preferred to all machines $r_{t_b}$ for $1 \leq b \leq |B|-1$).  Thus, in all cases, the machine $r_m$ is strictly preferred
to all other machines. Hence, $DR$ sets prices so that a rational job $j$ always chooses $r_m$ ($FF$ is free to assign $j$ to $r_m$).  Moreover, after job $j$ chooses $r_m$,
$DR$ checks if $S = \emptyset$ (which it is in this case, as $T = \emptyset$), and updates the estimate $\Lambda$ along with all virtual loads
in the same manner as $FF$.
\end{proof}

\noindent {\bf Proof of Lemma~\ref{lem:TnemptySempty}:}

\begin{proof}
Now we consider the case when set $T$ is nonempty, but set $S$ is empty.  In this case, the algorithm $FF$ is allowed to assign job $j$
in several ways. $FF$ is free to assign job $j$ to any machine of the fastest speed and begin a new phase (i.e.,
update the estimate $\Lambda$ and set all virtual loads to $0$).  $FF$ may also choose to forgo starting a new phase,
in which case it is free to assign job $j$ to any machine $r_i$ where $i \in T$.  Note that, in general, $FF$ is free
to assign job $j$ to any machine $r_i$ where $i \in T$ and $s_i \leq s_k$, but in this case $k = m$ (recall that $k$ is the minimum
machine index in $S$, but since $S = \emptyset$, $FF$ sets $k = m$).  $DR$ may or may not set prices so that job $j$ chooses a machine
of speed $s_m$ (namely, the representative $r_m$).  If job $j$ does choose a machine of speed $s_m$, then after
choosing the machine, $DR$ checks if $S = \emptyset$ (which it is in this case), and begins a new phase by updating the
estimate $\Lambda$ and resetting all virtual loads to $0$ (in a manner consistent with $FF$).

Hence, we need only consider the case when job $j$ chooses the representative of a machine with speed strictly less than $s_m$.
Suppose job $j$ chooses some machine $r_{t_b}$, where $1 \leq b \leq |B|$.  Notice that, if job $j$ chooses $r_{t_b}$ where
$b < |B|$ (so that $|B| \geq 2$), then we have $c_{r_{t_b}j} \leq c_{r_{t_{b+1}}j}$, which by Lemma~\ref{lem:prices} implies
that $t_b \in T$.  On the other hand, if job $j$ chooses $r_{t_b}$ where $b = |B|$, then assuming $s_{t_{|B|}} \neq s_m$
(as otherwise we are done, since this would contradict the fact that $j$ chooses a machine with speed strictly less than $s_m$),
we again have $c_{r_{t_{|B|}}j} \leq c_{r_mj}$, which implies $t_{|B|} \in T$ (by Lemma~\ref{lem:prices}).  Hence, in all cases job $j$ chooses
the representative $r_i$ of some machine $i$, where $i \in T$ (assuming it is not assigned to a machine of speed $s_m$).
\end{proof}

\noindent {\bf Proof of Lemma~\ref{lem:Snempty}:}

\begin{proof}
We consider the case when both sets $T$ and $S$ are nonempty.  In this case,  $FF$ is free to assign job $j$
to any machine $r_i$ where $i \in T$ and $s_i \leq s_k$.  We argue that $DR$ sets prices so that $j$ chooses a machine in the same manner.
Since $S \neq \emptyset$, there exists a machine $i \in S$ where $i$ satisfies
$\hat{\ell}_i(j-1) + \frac{p_j}{s_i} \leq 2\cdot \Lambda$, which implies $p_j \leq s_i(2\cdot \Lambda - \hat{\ell}_i(j-1)) = \mu_{i_q}$
for some $1 \leq q \leq m$.  Notice that $\mu_{i_q} \leq \mu_{i_m}$, and hence we have $p_j \leq \mu_{i_m}$.  Moreover, we always
have the property $\mu_{t_{|B|}} = \mu_{i_m}$, since otherwise set $A$ would be nonempty and hence $DR$ would add more
elements to set $B$.  Therefore, we know $p_j \leq \mu_{t_{|B|}}$.  Thus, let $1 \leq b \leq |B|$ be the smallest value satisfying $p_j \leq \mu_{t_b}$
(notice that such a value $b$ must exist as $p_j \leq \mu_{t_{|B|}}$).

We now argue that $s_k = s_{t_b}$ and assume towards a contradiction that $s_k < s_{t_b}$ (clearly,
$s_k \leq s_{t_b}$, since machine $t_b$ is in $S$).  If $b=1$, we are done since $s_{t_1}$ is the speed of the slowest machine,
and hence $s_{t_1} \leq s_k$ (which implies $s_{t_1} = s_k$).  Suppose $b > 1$, in which case we have $\mu_{t_1} \leq \cdots \leq \mu_{t_{b-1}} < p_j$.
If a machine of speed $s_k$ was in set $A$ when $DR$ added machine $t_b$ to set $B$, then we have a contradiction as a machine
of speed $s_k$ would have been added to $B$ instead of machine $t_b$, since we assumed $s_k < s_{t_b}$.  Hence,
assume that all machines of speed $s_k$ were already removed from set $A$ when $DR$ added $t_b$ to $B$.
This means that all machines of speed $s_k$ appear earlier in the sorted ordering $\mu_{i_1} \leq \cdots \leq \mu_{i_m}$ than
machine $t_{b-1}$ (possibly including the same position).  This is a contradiction, since we know $\mu_{t_{b-1}} < p_j$,
and hence all machines of speed $s_k$ do not belong to set $S$, in which case every machine of minimum speed in $S$
has speed strictly more than $s_k$.

If $s_{t_{|B|}} \neq s_m$, then observe that job $j$ strictly prefers $r_{t_{|B|}}$ to $r_m$, since
$p_j \leq \mu_{t_b} \leq \cdots \leq \mu_{t_{|B|}}$ and
hence $p_j < s_{t_{|B|}}((2+\epsilon)\Lambda - \hat{\ell}_{t_{|B|}})$, implying $c_{r_{t_{|B|}}j} < c_{r_m j}$
(the proof of Lemma~\ref{lem:prices} shows this).
Thus, $DR$ sets prices so that $j$ chooses some machine $r_{t_h}$ where $1 \leq h \leq |B|$, whether or not $s_{t_{|B|}} = s_m$
(since if $s_{t_{|B|}} = s_m$, then $j$ always chooses some machine $r_{t_h}$ where $1 \leq h \leq |B|$).
Assuming $b < |B|$, then since $p_j \leq \mu_{t_b} \leq \mu_{t_{b+1}} \leq \cdots \leq \mu_{t_{|B|}}$
we have $t_b,\ldots,t_{|B|} \in T$, implying $c_{t_b j} \leq \cdots \leq c_{t_{|B|}j}$.
In fact, since $p_j \leq \mu_{t_b} < s_{t_b}((2+\epsilon)\Lambda - \hat{\ell}_{t_b})$ we know job $j$ strictly prefers $t_b$
to all other machines $t_{b+1},\ldots,t_{|B|}$ (by Lemma~\ref{lem:prices}).  Thus, job $j$ chooses some machine $r_{t_h}$ where
$1 \leq h \leq b$.  Notice that all such machines have $s_{r_{t_h}} = s_{t_h} \leq s_{t_b} = s_k$.
If job $j$ chooses $r_{t_b}$, then we are done since $p_j \leq \mu_{t_b}$, implying $t_b \in T$.
Otherwise, job $j$ chooses some machine $r_{t_h}$ for $1 \leq h < b$, which implies that
$c_{r_{t_h}j} \leq c_{r_{t_{h+1}}j}$, and hence $t_h \in T$.  Thus, $DR$ behaves
consistently with $FF$.
\end{proof}

\end{document}